\documentclass[times]{article}
\pdfoutput=1
\usepackage{fullpage}
\usepackage{natbib}
\usepackage{moreverb,url}
\usepackage{amsmath}
\usepackage{amsthm}

\usepackage{caption}
\usepackage[skip=0pt,justification=centering]{subcaption}

\usepackage[colorlinks,bookmarksopen,bookmarksnumbered,citecolor=red,urlcolor=red]{hyperref}
\usepackage{tikz}
\usetikzlibrary{calc,trees,positioning,arrows,fit}
\usetikzlibrary{shadows,chains,shapes.geometric,
    decorations.pathreplacing,decorations.pathmorphing,
    shapes,matrix,shapes.symbols}

\tikzstyle{materia}=[draw, fill=blue!20, text width=6.0em, text centered,
  minimum height=1.5em,drop shadow]
\tikzstyle{etape} = [materia, text width=4.5em, minimum width=6.5em,
  minimum height=3em, rounded corners, drop shadow]
\tikzstyle{betape} = [materia, text width=6em, minimum width=10em,
  minimum height=3em, rounded corners, drop shadow]
\tikzstyle{plaintext} = [draw=none, fill=none, text width=6em, minimum width=10em,
  minimum height=3em, font=\fontsize{100}{100}\selectfont]
\tikzstyle{line} = [draw, thick, -latex']

\newdimen\nodeDist
\newdimen\snodeDist
\tikzset{
    events/.style={ellipse, draw, align=center},
}

\newcommand\BibTeX{{\rmfamily B\kern-.05em \textsc{i\kern-.025em b}\kern-.08em
T\kern-.1667em\lower.7ex\hbox{E}\kern-.125emX}}

\begin{document}


\title{On Prediction and Tolerance Intervals\\ for Dynamic Treatment Regimes}

\author{Daniel J. Lizotte \\
Departments of Computer Science and Epidemiology \& Biostatistics,\\ The University of Western Ontario, London, Ontario, Canada
\and Arezoo Tahmasebi \\
Department of Epidemiology \& Biostatistics,\\ The University of Western Ontario, London, Ontario, Canada
}

\maketitle

\begin{abstract}
We develop and evaluate tolerance interval methods for dynamic treatment regimes (DTRs) that can provide more detailed prognostic information to patients who will follow an estimated optimal regime. Although the problem of constructing confidence intervals for DTRs has been extensively studied, prediction and tolerance intervals have received little attention. We begin by reviewing in detail different interval estimation and prediction methods and then adapting them to the DTR setting. We illustrate some of the challenges associated with tolerance interval estimation stemming from the fact that we do not typically have data that were generated from the estimated optimal regime. We give an extensive empirical evaluation of the methods and discussed several practical aspects of method choice, and we present an example application using data from a clinical trial. Finally, we discuss future directions within this important emerging area of DTR research.
\end{abstract}

\section{Introduction}

Dynamic Treatment Regimes (DTRs), also known as adaptive treatment strategies or treatment policies, are a key tool for providing data-driven sequential decision-making support. A DTR is a sequence of decision functions that take up-to-date patient information as input and produce a recommended treatment. Thus, a DTR is a mathematical representation of the sequential decision-making process. Using this representation, we can use previously collected decision-making data to estimate an ``optimal'' DTR, where optimality is most often defined in terms of expected outcome. That is, a DTR is optimal if it produces the best outcome, on average, over a patient population. We will use this definition of optimality throughout our work.

Each decision in an optimal DTR is made in the service of achieving maximal expected outcome. However, the outcome of any particular individual under an optimal regime may vary widely from this expectation. Indeed, DTRs have been applied in many very challenging areas of medicine, including psychiatry, cancer, and HIV, where patient outcomes are known to be highly variable, or, equivalently from our perspective, difficult to predict.

It is with this variability in mind that we consider different methods for assessing the variability in {\em individual outcomes} under a given DTR. Our objective is to quantify for the decision-maker not our certainty about the expectation of outcomes, but rather our uncertainty about what the observed outcome might be {\em for a particular patient.}

We begin by formally defining DTRs, and we review point and interval estimation techniques for relevant parameters of the optimal DTR. We then review definitions and existing methods for confidence intervals, prediction intervals, and tolerance intervals. Following this background, we formally describe our problem of interest in the context of using DTRs to provide decision support. 

We will see that the main technical challenge associated with constructing tolerance intervals for DTRs stems from not having a sample drawn from the correct distribution. Thus, our methods will use re-weighting and re-sampling to allow us to apply existing tolerance interval methods in this setting. To help illustrate the technical challenge, we first describe a na\"{i}ve strategy for constructing tolerance intervals whose performance is poor, and we then present two novel strategies for constructing valid tolerance intervals for the response under a given dynamic treatment regime. We present an empirical evaluation of the methods, and we conclude by discussing their implications and directions for future work.

\section{Background}\label{ss:background}

In the following, we review basic concepts pertaining to DTRs, the estimation of optimal regimes, and concepts and issues surrounding interval estimation and prediction.

\subsection{Dynamic Treatment Regimes}
\newcommand{\E}{\mathrm{E}}
\newcommand{\X}{\mathcal{X}}
\newcommand{\Y}{\mathcal{Y}}
\renewcommand{\P}{\mathrm{Pr}}
\newcommand{\T}{\mathsf{T}}

DTRs are a mathematical formalism meant to capture the decision-making cycle of information gathering, followed by treatment choice, followed by outcome evaluation. They have been defined at different levels of generality by many authors  \citep{Schulte2014,laber14dynamic,laber14setvalued,lizotte12linear,Nahum-Shani2012,Nahum-Shani2012a,lizotte10multiple,shortreed11mlj}. 
Here, we focus on regimes with two decision points; thus for this work we consider a DTR to be a sequence of two functions $(\pi_1,\pi_2)$ which map up-to-date patient information at the first and second decision points, respectively, to distributions over the space of available treatments at each decision point. We represent the information (covariates) about a given patient at point $t$ by $s_t$, which we view as a realization of a random variable $S_t$. Similarly, we denote the chosen treatment (action) by $a_t$, which is a realization of $A_t$. For a patient who follows a DTR $(\pi_1,\pi_2)$, we will have $A_1 \sim \pi_1(s_1)$ and $A_2 \sim \pi_2(s_1,a_1,s_2)$. We let $y$ be the observed outcome or {\em reward} attained by a patient after following a regime, and we follow the convention that larger values of $y$ are preferable. For a patient following a given regime, we observe $(s_1, a_1, s_2, a_2, y)$, the {\em trajectory} for that patient. 

Trajectory data may come from various observational and experimental sources, for example from Sequential Multiple Assignment Randomized Trials (SMARTs) \citep{Nahum-Shani2012,Nahum-Shani2012a,Collins2014}. A SMART is an experimental design under which patients follow a DTR that applies randomly assigned treatments. We will call such a DTR an {\em exploration} DTR or {\em exploration policy.} The goal of running a SMART is analogous to that of running a pragmatic randomized controlled trial---to evaluate the comparative effectiveness of different treatment options in an unbiased way. This comparative effectiveness information can then be used to estimate an {\em optimal} DTR. An optimal DTR is a pair of decision functions $(\pi_1,\pi_2)$ that maximize $\E[Y|S_1,A_1,S_2,A_2;\pi_1,\pi_2]$ where $A_t \sim \pi_t(S_t)$. Thus, an optimal DTR produces maximal expected outcome when applied to a population of patients. In this work, we focus on the setting where the {\em exploration} DTR is stochastic, but the candidate optimal DTRs under consideration are deterministic.

\subsection{$Q$-learning}
\newcommand{\Qhat}{\hat{Q}}
\newcommand{\pihat}{\hat{\pi}}
\newcommand{\is}{\mathord{=}}
Several methods are available for estimating an optimal DTR from data collected under an exploration DTR. Here, we review one such method called {\em $Q$-learning} \citep{Schulte2014,Huang2015}. $Q$-learning works by estimating $Q$ functions ($Q$ for ``quality'') that predict expected outcome given current covariates and treatment choice. In our 2-decision point setting, we have
\begin{equation*}
Q_2(s_1,a_1,s_2,a_2) = \E[Y|S_1\is s_1, A_1\is a_1, S_2\is s_2,A_2\is a_2].
\end{equation*}
Note that unlike the expectation in the previous section which averages over patients, $Q_2$ gives the expectation of $Y$ {\em conditioned on particular patient observations and treatment choices.}
The definition of $Q_2$ implies an optimal decision function $\pi^*_2(s_1,a_1,s_2) = \arg\max_{a'_2} Q_2(s_1,a_1,s_2,a'_2)$. $Q_2$ can be estimated using any regression method. Having obtained an estimate $\Qhat_2$ of $Q_2$, our estimate of the optimal second decision function is $\pihat^*_2(s_1,a_1,s_2) = \arg\max_{a'_2} \Qhat_2(s_1,a_1,s_2,a'_2)$. 

The optimal $Q$-function for the first decision point produces the conditional mean of $Y$ given $S_1$ and $A_1$ and {\em given that the optimal decision function $\pi^*_2$ is used at the second decision point.} In $Q$-learning, we estimate $Q_1$ by
\begin{equation*}
\Qhat_1(s_1,a_1) \approx \E[\max_{a'_2} \Qhat_2(s_1,a_1,S_2,a'_2)
|S_1\is s_1, A_1\is a_1]
\end{equation*}
where the expectation is over $S_2$ conditioned on $S_1$ and $A_1$. The quantity $\max_{a'_2} \Qhat_2(s_1,a_1,S_2,a'_2)$ is sometimes called the {\em pseudooutcome}, and is denoted $\tilde y$. In order to estimate $Q_1$, we compute the pseudooutcome for each trajectory in our dataset, and then regress them on $S_1$ and $A_1$ to estimate $Q_1$. Again, any regression method can be used to estimate $Q_1$, in principle. Our corresponding estimate of the optimal first decision function is then  $\pihat^*_1(s_1) = \arg\max_{a'_1} \Qhat_1(s_1,a'_1)$, and our estimate of the optimal DTR is $(\pihat^*_1,\pihat^*_2)$. Note that this DTR is deterministic.

We focus on $Q$-learning in this work, but several other methods are available for estimating optimal DTRs, including $A$-learning \citep{Blatt2004,Schulte2014}, the closely-related $g$-estimation \citep{Moodie2009,Orellana2010,Barrett2014}, and direct policy search \citep{Zhao2014,Zhao2015}.

\subsection{Interval Estimation}

For consistency, in the following we use $y$s to represent observed outcomes, and $x$s to represent covariates, even in non-regression settings.

\subsubsection{Confidence Intervals}
\newcommand{\CIl}{{\ell_c}}
\newcommand{\CIu}{{u_c}}
\newcommand{\CI}{{(\CIl,\CIu)}}

A {\em confidence interval} $\CI$ with level $1 - \alpha$ for a parameter $\theta$ is a functional of a dataset $\Y = \{y_1, ..., y_n\}$ of realizations of a random variable $Y$, with the property that
\begin{equation}\label{eq:CIprob}
\Pr[\theta \in \CI] \ge 1 - \alpha.
\end{equation}
The probability statement (\ref{eq:CIprob}) is over datasets containing i.i.d.\ samples of $Y$. The goal of a confidence interval is to provide confidence information about the estimated location of an underlying distributional parameter. Though not our main focus, confidence intervals are by far the most well-known class of interval estimates, and they are closely related to the prediction and tolerance intervals we will develop and investigate.

\subsubsection{Prediction Intervals}
\newcommand{\PIl}{{\ell_p}}
\newcommand{\PIu}{{u_p}}
\newcommand{\PI}{{(\PIl,\PIu)}}
\newcommand{\new}{\mathrm{new}}

A {\em prediction interval} $\PI$ with level $1 - \alpha$ is a functional of a dataset $\Y = \{y_1, ..., y_n\}$ of realizations of a random variable $Y$, with the property that
\begin{equation}\label{eq:PIprob}
\Pr[Y_{\new} \in \PI] \ge 1 - \alpha.
\end{equation}
Here, $Y_{\new}$ represents a single future observation that was not contained in the original data $\Y$. The goal of a prediction interval is to provide confidence information about where this new observation might fall. However, we note, as others have \citep{Vardeman2012}, that there is often confusion surrounding the probability statement (\ref{eq:PIprob}). In particular, the statement is over the joint distribution of $Y_1,...,Y_{n},Y_\new$. A prediction interval formed from a dataset traps {\em one additional observation} with probability $1 - \alpha$. It offers no guarantees about trapping more than one additional observation, and indeed no guarantees regarding our confidence in the {\em content} of an interval, that is, of the quantity $F_Y(\PIu) - F_Y(\PIl)$ where $F_Y$ is the cumulative distribution function of $Y$. (For example, a prediction interval that has content $1.0$ half the time and content $0.9$ half the time has property (\ref{eq:PIprob}) for $\alpha=0.05$, as does an interval that always has content $0.95$.)

\newcommand{\N}{\mathcal{N}}

The well-known normal-theory prediction interval for $Y$ 
\citep{NeterJohnandWassermanWilliamandKutner1989} is given by
\begin{equation}
\PI_{\N} = \bar{y} \pm t_{\alpha/2;n-1}\hat{\sigma}_Y \sqrt{1 + \frac{1}{n}} \label{eq:normPI}
\end{equation}
where $\bar{y}$ is the sample mean, $\hat{\sigma}_Y$ the sample standard deviation, and $t_{\alpha/2;n-1}$ is the $\alpha/2$ quantile of a $t$-distribution with $n - 1$ degrees of freedom. Note that the   validity of (\ref{eq:normPI}) is predicated on normality of $Y$, regardless of sample size.

The corresponding prediction interval for $Y|X\is x$ in the linear regression setting on $p$ parameters is
\begin{equation}\label{eq:normRegPI}
\PI_\N = \hat{y} \pm t_{\alpha/2;n-p} \hat{\sigma}_{Y|X\is x} \sqrt{1 + x^\T (\mathrm{X}^\T\mathrm{X})^{-1} x}
\end{equation}
where $x$ represents the location of a new sample, $\hat{y}$ is the prediction of $\E[Y|X\is x]$, $\hat{\sigma}_{Y|X\is x}$ is the sample standard deviation of the residuals, $\mathrm{X}$ is the design matrix for the regression, and $t_{\alpha/2;n-p}$ is the $\alpha/2$ quantile of a $t$-distribution with $n - p$ degrees of freedom. Equation (\ref{eq:normRegPI}) is predicated on the normality of $Y|X=x$ and on homoscedasticity of the residuals.
\subsubsection{Tolerance Intervals}
\newcommand{\TIl}{{\ell_t}}
\newcommand{\TIu}{{u_t}}
\newcommand{\TI}{{(\TIl,\TIu)}}

A {\em tolerance interval} $\TI$ with level $1 - \alpha$ and content $\gamma$ is also a functional of a dataset $\Y = \{y_1, ..., y_n\}$. It has the property that
\begin{equation}\label{eq:TIprob}
\Pr[F_Y(\TIu) - F_Y(\TIl) \ge \gamma] \ge 1 - \alpha.
\end{equation}
where $F_Y$ is the cumulative distribution function of $Y$.
Thus, a tolerance interval formed from a dataset traps at least $\gamma$ of the probability content of $Y$ with probability $1 - \alpha$, where the $1 - \alpha$ probability is over datasets.

One well-known normal theory approximate tolerance interval for $Y$ with confidence $1 - \alpha$ and content $\gamma$ is given by \cite{KrishnamoorthyKalimuthuandMathew2009} as 
\begin{equation}\label{eq:normTI}
\TI_\N = \bar{y} \mp \hat{\sigma}_Y \sqrt{\frac{(n - 1) \chi^2_{\gamma;1,1/n}}{\chi^2_{\alpha;n - 1}}}
\end{equation}
where $\bar{y}$ is the sample mean, $\chi^2_{\gamma;1,1/n}$ is the $\gamma$ quantile of a non-central $\chi^2$ distribution with 1 degree of freedom and noncentrality parameter $1/n$, and $\chi^2_{\alpha;n-1}$ is the $\alpha$ quantile of a $\chi^2$ with $n-1$ degrees of freedom. 

The corresponding tolerance interval for $Y|X\is x$ in the linear regression setting on $p$ parameters is \citep{Young2013}
\begin{equation}\label{eq:normRegTI}
\TI_\N = \hat{y} \mp \hat{\sigma}_{Y|X\is x} \sqrt{\frac{(n - p) \chi^2_{\gamma;1,1/n^*}}{\chi^2_{\alpha;n - p}}}
\end{equation}
where $\hat{y}$ is the prediction of $\E[Y|X\is x]$, $\hat{\sigma}_{Y|X\is x}$ is the sample standard deviation of the residuals, 
$n^* = \hat{\sigma}^2_{Y|X\is x} /\hat{\sigma}^2_{\hat{y}}$ is Wallis' ``effective number of observations" (\citeyear{Wallis1951}), and $\hat{\sigma}_{\hat{y}}$ is the standard error of $\hat{y}|X\is x$. Again, validity of (\ref{eq:normTI}) and (\ref{eq:normRegTI}) is predicated on the normality of $Y$ and $Y|X\is x$, respectively; (\ref{eq:normRegTI}) is also predicated on homoscedasticity.

\cite{Wilks1941} proposed a non-parametric tolerance interval that assumes only continuity of $F_Y$. The interval is given by the sample values corresponding to the minimum and maximum ranks $r$ for which
\begin{equation}
(1 - F_{\mathrm{Beta}}(\gamma; n - 2r + 1, 2r)) > 1 - \alpha \label{eq:wilksRanks}
\end{equation}
where $F_{\mathrm{Beta}}$ is the beta cumulative distribution function. Thus, the interval is constructed simply by truncating the sample to the ranks satisfying (\ref{eq:wilksRanks}), and then taking the minimum and maximum of the truncated sample to be the lower and upper limits of the tolerance interval, respectively.

\section{DTRs for Decision Support}

DTRs are an ideal formalism for providing data-driven {\em decision support}. The most basic approach to providing decision support would be to estimate an optimal DTR from SMART data, and then provide the estimated DTR $(\pihat^*_1, \pihat^*_2)$ to a decision maker, perhaps as a computer-based tool that produces the estimated optimal treatment by using current patient information as input to the previously estimated DTR.

Early in the development of DTRs it was recognized that this approach is problematic because it provides no confidence information about our recommendations. Just as we would not recommend one treatment over another if no statistically significant difference were obtained from a standard randomized controlled trial (RCT), neither should we recommend a single treatment in a DTR if in fact the alternatives are not known to be inferior with high confidence. This led to the development of confidence interval methods for the difference in mean expected outcome under different treatment choices within a regime \citep{Chakraborty2010,Chakraborty2013,Chakraborty2013a,laber14dynamic,Chakraborty2014}.

Such intervals can give us confidence that if we do recommend a single treatment, that treatment will provide a better outcome, in expectation over patients. However, they do not provide any information about what the range of possible outcomes might actually be for an {\em individual} patient. In particular, large SMARTs with 100s to 1000s of patients may discover statistically significant differences in mean outcome even when the effect sizes are small to moderate and variance in outcomes is still substantial. If this is the case, it may be better to avoid recommending a single treatment, or at least to provide more nuanced information about what the patient's experience is likely to be under the different treatment options.

In this work, we consider tolerance intervals as one method for providing this information. For a patient with $S_1 = s_1$ at the first decision point, rather than recommending treatment $\pihat^*_1(s_1)$ (even if it is statistically significantly better than the alternative in terms of mean outcome) we would present tolerance intervals for the outcome $Y$ under each possible action, and allow the decision-maker (or the patient-clinician dyad, in the context of patient-centred care \citep{Barry2012}) to decide on treatment based on the range of probable outcomes indicated by the intervals. For each interval, we condition on the observed $s_1$, the hypothetical $a_1$, and the estimated optimal regime $\pihat^*_2$ for the second stage.

Thus, we will construct tolerance intervals for $Y|S_1=s_1,A_1=a_1;\pi_2 = \pihat^*_2$, marginal over $S_2$ (whose distribution is governed by $S_1$ and $A_1$) and $A_2$ (whose distribution is governed by $A_1$, $S_1$, $S_2$ and $\pi_2$.) To do so, we will adapt several standard methods because typically {\em we do not have observations drawn from this distribution.} This is because, as we noted above, data from SMART studies and similar sources are generated according to an exploration DTR $(\pi^0_1,\pi^0_2)$, rather than according to an estimated optimal DTR $(\pihat^*_1,\pihat^*_2)$. 

\subsection{Aside: Non-regularity}

It is well-known that many kinds of inference on the parameters of an estimated optimal dynamic treatment regime, including confidence intervals, are plagued by issues of {\em non-regularity}  \citep{laber14dynamic}. Briefly, non-regularity is a result of the sampling distributions of corresponding estimators changing abruptly as a function of the true underlying parameters. It can lead to bias in estimates and anti-conservatism in inference. In dynamic treatment regimes, non-regularity occurs and inference is problematic when two or more treatments produce (nearly) the same mean optimal outcome. In this work, we will not specifically develop methods that are robust to non-regularity. This is because even in the absence of non-regularity, i.e.\ when optimal $Q$ values are well-separated from sub-optimal ones, there is significant variability in the performance of ``standard'' tolerance interval methods that is worthy of exploration and analysis. We will return to this point in the Discussion.

\section{Methods}\label{ss:methods}

We now detail our strategies for constructing tolerance intervals for  $Y|S_1\is s_1,A_1\is a_1;\pi_2\is \pihat^*_2$. As we mentioned above, the fundamental challenge of constructing intervals for this quantity is that in general we do not have samples drawn from this distribution---otherwise, we could use off-the-shelf tolerance interval methods. Note that we {\em can} use off-the-shelf methods for tolerance intervals for $Y|S_2\is s_2, A_2\is a_2$, because there is no need to account for future decision-making in that case; thus our work focuses on the first decision point. We begin by presenting a na\"{i}ve approach to constructing tolerance intervals that helps illustrate the main technical challenge to be addressed, and then we present our two proposed strategies: {\em inverse probability weighting}, and {\em residual borrowing}.

\subsection{Na\"{i}ve $Q$-Learning Tolerance Intervals}

Standard $Q$-learning involves estimating $Q_1(s_1,a_1)$, which predicts the expected $Y$ under the optimal regime. However, it does so using the pseudooutcome $\tilde{Y} = \max_{a'_2} \Qhat_2(s_1,a_1,s_2,a'_2)$ as the regression target, rather than the observed $Y$. Since the pseudooutcome targets are themselves predicted conditional means of $Y$, they carry no variance information about $Y|S_2,S_1,A_1$ under the estimated optimal policy, even among trajectories that (by chance) followed the estimated optimal policy. To see this, suppose that we had several trajectories, all of which had the same $s_1,a_1,s_2,a_2$, and all of whom happened to follow the estimated optimal policy. Even though their {\em observed} outcomes $y$ might have all been different, simply due to unexplained (but still important) variation in $Y$, they would all be assigned the same pseudooutcome value, and the sample variance of the pseudooutcomes in this group is zero.

This observation highlights the key aspect of $Q$-learning and related methods that precludes direct estimation of variability in $Y$. Dynamic programming methods for estimating conditional means of sequential outcomes can ``throw away'' residual variance without negative repercussions when backing up values, essentially because of the law of total expectation. The benefit of this approach is a reduction in the variance of $Q$ estimates by allowing the use of the entire dataset of trajectories for estimating $Q$-functions for earlier decision points. The drawback is that such methods cannot directly estimate other distributional properties of $Y$, including variance and higher-order moments, quantiles, and so on.

If most of the variability in $Y$ were explained by $S_2$ and $A_2$---that is, if the variance of $Y|S_2,A_2$ were nearly zero---we might be able to construct approximate tolerance intervals for $Y$ by constructing parametric tolerance intervals for the pseudooutcome, for example using (\ref{eq:normRegTI}). In the case of a saturated model with discrete $S_1$ and $A_1$, we could construct non-parametric tolerance intervals for each pattern of $(s_1,a_1)$ using the pseudooutcome with (\ref{eq:wilksRanks}). However, as expected, will see in our empirical results that this approach is not very effective if in fact the variance of $Y|S_2,A_2$ is not near zero.

\subsection{Inverse Probability Weighting}

One approach to obtaining variance information about $Y$ under $\pihat_2^*$ is to select from our dataset only those trajectories whose second-stage treatment matches what $\pihat_2^*$ would have assigned, i.e., the trajectories $(s_1,a_1,s_2,a_2,y)$ for which $a_2 = \pihat_2^*(s_1,a_1,s_2)$. This subset contains all of the trajectories that have positive probability under the estimated DTR.

\newcommand{\dom}{\mathbf{dom}\;}
Consider a joint distribution over $S_2,A_2,\Pi,A_2^0,A_2^*,M,Y$ conditioned on $S_1$ and $A_1$. (All statements in the remainder of this subsection are implicitly conditioned on $S_1$ and $A_1$; explicitly maintaining this is too cumbersome.) Here, $A_2^0$ is the action chosen by $\pi_2^0$, and $A_2^*$ is the action chosen by $\pihat^*_2$, which is assumed to be deterministic given $S_2$. Let $M$ (for {\em match}\footnote{Note we are not matching trajectories with other trajectories---we are identifying trajectories whose action matches a DTR of interest.}) be $1$ if $A_2^* = A_2^0$, or $0$ otherwise. Let $\Pi$ be binary, and define $A_2$ such that $A_2 = A_2^0$ if $\Pi = 0$ and $A_2 = A_2^*$ if $\Pi = 1$. The dependencies among all of these variables are illustrated in Figure~\ref{fig:deps} using a directed graphical model \citep{Koller2009}. 

The distribution of $Y$ among matched trajectories is governed by $Y|\Pi=0,M=1$. The distribution of $Y$ among trajectories gathered using $\pihat^*_2$ is $Y|\Pi = 1$.
Note that while the distribution of $Y|S_2,A_2,\Pi\is 0,M\is 1$ is identical to the distribution of $Y|S_2,A_2,\Pi\is 1$ due to the conditional independence structure, the distribution of $Y|\Pi=0,M=1$ may be different from $Y|\Pi = 1$ if there is dependence of $M$ on $S_2$. We describe this phenomenon using the following lemma.
\clearpage

\begin{figure}[!h]
\centering
\begin{tikzpicture}[>=stealth']
\node [events] (state) {$S_2$};
\node [events, above right = 0.1cm and 1cm of state](a0){$A_2^0$};
\node [events, below right = 0.1cm and 1.5cm of state](astar){$A_2^*$};
\node [events, below right = 0.1cm and 1cm of a0] (match) {$M$};
\node [events, below right = 0.7cm and 0.1cm of state](a){$A_2$};
\node [events, below = 2cm of state] (y) {$Y$};
\node [events, below = 0.5cm of astar] (pi) {$\Pi$};

\draw [->] (state) -- (a0);
\draw [->] (state) -- (astar);
\draw [->] (a0) -- (match);
\draw [->] (astar) -- (match);
\draw [->] (state) -- (y);
\draw [->] (a0) -- (a);
\draw [->] (astar) -- (a);
\draw [->] (pi) -- (a);
\draw [->] (a) -- (y);
\end{tikzpicture}
\caption{Graphical model depicting the dependence structure of $S_2,A_2^0,A_2^*,A_2,\Pi,M,Y|S_1,A_1$. Note that the structure is the same for all values of $S_1$ and $A_1$.} \label{fig:deps}
\end{figure}
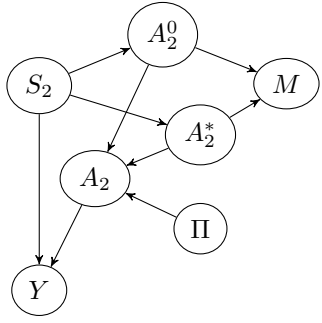

\newtheorem{lemma}{Lemma}

\begin{lemma}\label{lem:ratio}
Let $S_2,A_2^0,A_2^*,A_2,\Pi,M,Y$ be defined as above, and assume $\Pr(S_2) > 0 \implies \Pr(S_2|M\is 1) > 0$. Then
\begin{equation*}
\Pr(Y|\Pi\is 1) = \sum_{S_2} \frac{\Pr(S_2)}{\Pr(S_2|M\is 1)} \Pr(Y,S_2|\Pi \is 0, M\is 1).
\end{equation*}
\end{lemma}
\begin{proof}
In the following, we abuse notation by allowing $\Pr$ to represent a probability or a density, as appropriate, and we allow $\sum$ to indicate a sum or an integral. The message in any case remains the same.

First we note that
\begin{equation}
\Pr(Y|\Pi\is 0,M\is 1) 
 = \sum_{S_2,A_2} \Pr(Y,S_2,A_2|\Pi\is 0,M\is 1)
 = \sum_{S_2,A_2} \left\{
\begin{aligned}
\Pr(Y|S_2,A_2,\Pi\is 0,M\is 1)\cdot \\
\Pr(A_2|S_2,\Pi \is 0, M\is 1)\cdot \\
\Pr(S_2|\Pi\is 0, M\is 1)
\end{aligned}
\right\}.\label{eq:YgP0M1}
\end{equation}

The data generating distribution under $\pihat^*_2$ is
\begin{align}
\Pr&(Y|\Pi\is 1) \nonumber \\
&= \sum_{S_2,A_2} \Pr(Y,S_2,A_2|\Pi\is 1)\nonumber \\
&= \sum_{S_2,A_2} \Pr(Y|S_2,A_2,\Pi\is 1) \Pr(S_2,A_2|\Pi\is 1)\nonumber \\
&= \sum_{S_2,A_2} \Pr(Y|S_2,A_2,\Pi\is 0,M\is 1) \Pr(S_2,A_2|\Pi\is 1)\label{eq:YgP1}
\end{align}
where the last step follows from conditional independence of $Y$ and $(\Pi,M)$ given $S_2$ and $A_2$. Furthermore, 
\begin{align}
\Pr&(S_2,A_2|\Pi\is 1) \nonumber \\
&= \Pr(A_2|S_2,\Pi\is 1) \Pr(S_2|\Pi\is 1)\nonumber \\
&= \Pr(A_2|S_2,\Pi \is 0, M\is 1) \Pr(S_2|\Pi\is 1)\nonumber \\
&= \Pr(A_2|S_2,\Pi \is 0, M\is 1) \Pr(S_2|\Pi\is 0) \label{eq:S2A2gP1}
\end{align}
where the second step follows because $A_2^*$ is deterministic given $S_2$\footnote{This assumption is critical: if $A^*_0|S_2$ is not deterministic, the relationship between $Y|\Pi\is 1$ and $Y|\Pi\is 0,M\is 1$ is more complicated.} and from the definition of $\Pi$ and $M$, and the third step follows from independence of $S_2$ and $\Pi$. By combining  (\ref{eq:YgP1}) and (\ref{eq:S2A2gP1}) and comparing with (\ref{eq:YgP0M1}), we obtain
\begin{align*}
\Pr(Y|\Pi\is 1) & = \sum_{S_2,A_2} \left\{
\begin{aligned}
\Pr(Y|S_2,A_2,\Pi\is 0,M\is 1)\cdot \\
\Pr(A_2|S_2,\Pi \is 0, M\is 1)\cdot \\
\Pr(S_2|\Pi\is 0)
\end{aligned}
\right\} \\
= & \sum_{S_2} \left\{
\begin{aligned}
\Pr(Y|S_2,\Pi\is 0,M\is 1)\cdot \\
\Pr(S_2|\Pi\is 0)
\end{aligned}
\right\} \\
= & \sum_{S_2} \frac{\Pr(S_2|\Pi\is 0)}{\Pr(S_2|\Pi\is 0,M\is 1)} \Pr(Y,S_2|\Pi \is 0, M\is 1)\\
= & \sum_{S_2} \frac{\Pr(S_2)}{\Pr(S_2|M\is 1)} \Pr(Y,S_2|\Pi \is 0, M\is 1)
\end{align*}
where the final step is by independence of $S_2$ and $\Pi$.
\end{proof}
\newtheorem{corollary}{Corollary}
\begin{corollary}
If $S_2$ and $M$ are independent, then $Y|\Pi\is 0,M\is 1$ has the same distribution as $Y|\Pi\is 1$.
\end{corollary}
\begin{proof}
Follows immediately from Lemma~\ref{lem:ratio}.
\end{proof}

To achieve independence of $S_2$ and $M$, we could ensure during data collection that $A_2^0$ is independent of $S_2$, which in turn can be achieved by equal randomization independent of $S_2$. This is common, but not universal, in SMART designs \citep{Collins2014}. If $A_2^0 | S_2\is s_2 \sim \mathrm{Bernoulli}(\theta^0)$ and $A_2^* | S_2\is s_2 \sim \mathrm{Bernoulli}(\theta^*_{s_2})$, then 
\begin{align*}
\Pr(M\is 1|S_2\is s_2) & = \theta^0 \theta^*_{s_2} + (1 - \theta^0)(1 - \theta^*_{s_2})\\
& = 1 - \theta^0 + \theta^*_{s_2}(2\theta^0 - 1).
\end{align*}
Hence, if $\theta^0 = 0.5$, then $\Pr(M\is 1|S_2) = \Pr(M\is 1) = 0.5$, and $\Pr(S_2|M\is 1) = \Pr(S_2)$. Using this subset of trajectories whose $s_2$ matches $\pihat^*_2(s_2)$, we can regress $Y$ on $S_1$ and $A_1$ to construct tolerance intervals using (\ref{eq:normRegTI}), or, as above, we can construct non-parametric tolerance intervals for each pattern of $(s_1,a_1)$ using ($\ref{eq:wilksRanks}$).

\begin{figure*}
\includegraphics[width=\textwidth]{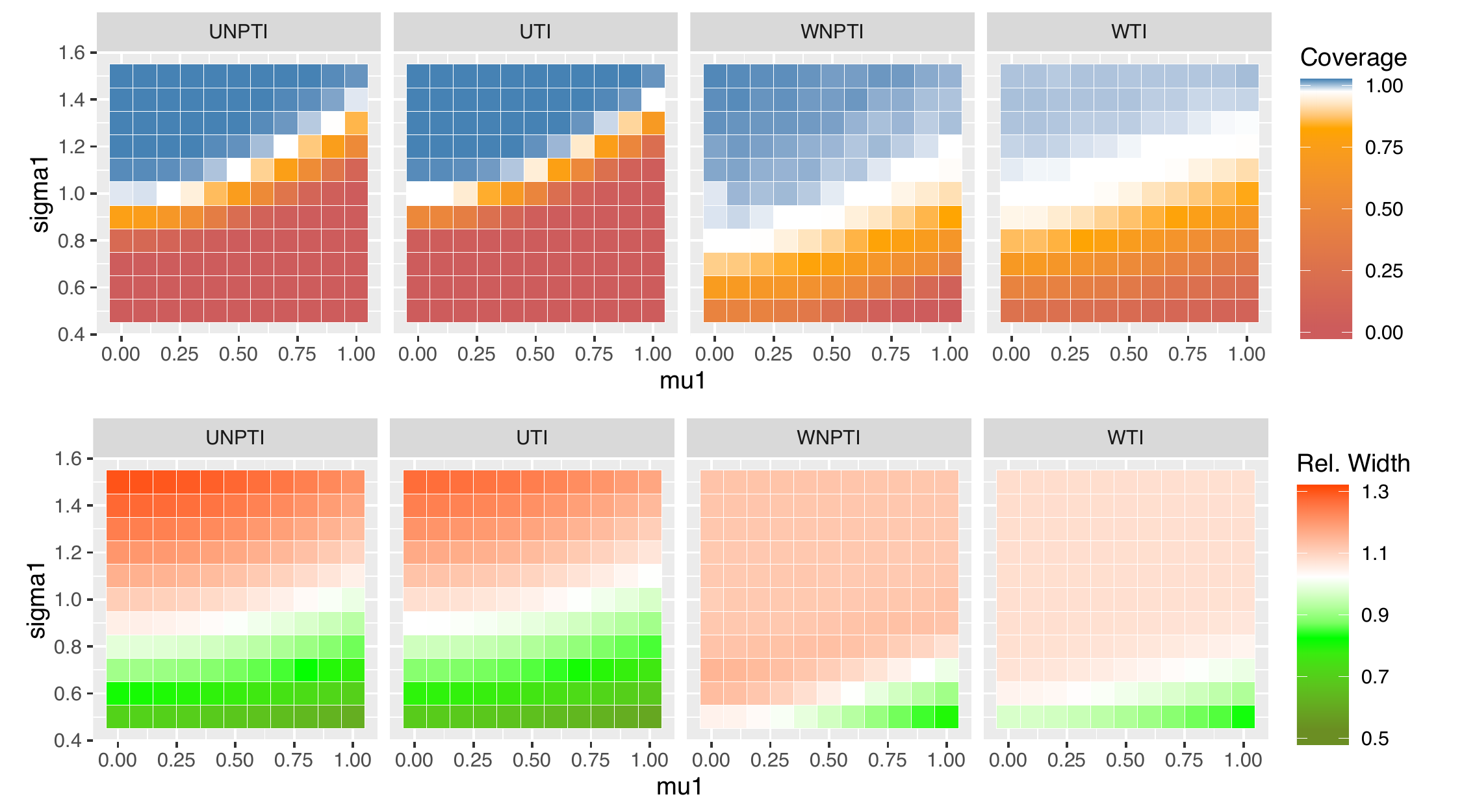}
\caption{Comparison of coverage and width of inverse probability weighted tolerance interval methods. Axes represent a space of simple generative models. Lighter colouration indicates better performance.\label{fig:meanTIs}}
\end{figure*}

Dependence of $M$ on $S_2$ is problematic because of the effect of $S_2$ on $Y$. When $M$ depends on $S_2$, conditioning on $M$ can affect the distribution of $Y$ through $S_2$, meaning that the distribution of $Y|S_1,A_1,\Pi\is 0,M \is 1$ we estimate by collecting data under $\pi^0_2$ is not what we would have obtained had we collected data under $\pihat^*_2$ and ignored (i.e. marginalized over) $M$.

To correct the problem of the distribution of $S_2|S_1,A_1$ among the matched trajectories, we employ {\em inverse probability weighting.} To do so, we construct a propensity score model, not for the probability of treatment, but for the {\em probability of following the estimated optimal DTR}, i.e.\ $\Pr(M\is 1|S_2,S_1,A_1)$. Using this model, we can then re-weight the trajectories so that the distribution of $S_2|M\is 1,S_1,A_1$ matches the distribution of $S_2|S_1,A_1$ as well as possible. The weight function is therefore
\begin{equation}\label{eq:impweights}
w(s_1,a_1,s_2) = \frac{\Pr(S_2\is s_2|S_1\is s_1,A_1\is a_1)}
{\Pr(S_2\is s_2|S_1\is s_1,A_1\is a_1,M\is 1)}.
\end{equation}
These are sometimes known as {\em importance weights}. We note that in causal inference, importance weights are sometimes used to adjust for an association between the probability of receiving treatment and the observed outcome. Here, they are used to adjust for an association between the probability of following the estimated optimal policy and the observed outcome through the variable $S_2$. Note that estimating the two densities in (\ref{eq:impweights}) separately is not necessary to estimate the function $w$; it can be estimated using any density ratio estimation method. Logistic regression is one common approach but many others are available. In related weighting methods for causal inference, practitioners have found that a flexible model for $w$ is often preferable to a simpler one \citep{Ghosh2011}. 

To use the weighted data for building tolerance intervals, we must adapt existing methods for use with the weights. To build normal-theory regression tolerance intervals using the weighted data, we first estimate $\hat{y}|X\is x$ using weighted least squares. We then use the resulting mean estimate, together with a weight-based sandwich estimate of $\hat{\sigma}_{\hat y}$ to construct the tolerance interval as per (\ref{eq:normRegTI}). To build non-parametric tolerance intervals, we obtain weighted estimates of the ranks obtained by linear interpolation of the weighted empirical distribution \citep{Harrell2015}. We then construct the Wilks interval as per (\ref{eq:wilksRanks}).

\newcommand{\bern}{\mathrm{Bernoulli}}
Figure~\ref{fig:meanTIs} shows the empirical results of applying weighted tolerance intervals in a simple scenario. Our goal here is to verify that the weighting scheme can counteract some of the dependence on $M$. (We will evaluate them more fully in the next section.) The data are drawn from a two-variable generative model with $M \sim \bern (0.5)$ and $Y|M\is m \sim \N(\mu_m,\sigma_m)$. Our goal is to produce a tolerance interval for $Y$, marginal over $M$, using only data for which $M\is 1$. The sample size for $M=1$ was $n = 500$, and the weights were computed analytically. Parameters for $Y|M \is 0$ were fixed at $\mu_0 = 0$ and $\sigma_0 = 1$. Parameters for $Y|M \is 1$ were varied to illustrate how performance of the weighted tolerance intervals changed as the distribution of $Y|M\is 1$ deviated from the marginal distribution of $Y$. The top row of heatmaps shows the coverage of each method, that is, the proportion of times out of 1000 Monte Carlo replicates for which the computed tolerance interval had at least $\gamma = 0.9$ probability content. The confidence level $1-\alpha$ was set to $0.95$; in the plot, Monte Carlo coverages that are not statistically significantly different from 0.95 are coloured pure white. Over-coverage is coloured blue, and under-coverage is coloured orange. The second row plots the average width of the tolerance intervals, normalized by the width of the optimal tolerance interval constructed from the true quantiles of $Y$, with unit relative width coloured white.

Methods beginning with U are unweighted, and methods beginning with W are weighted. Methods containing NP are nonparametric, and those without NP are normal-theory. (Table \ref{tab:legend} gives the complete key to the method names.) Note that except when $\mu_1 = 0$ and $\sigma_1 = 1$, $Y$ is nonnormal. As one would expect, performance when $\mu_1=0,\sigma_1=1$ is very good across all methods; in this case, $\Pr(Y|M=1) = \Pr(Y)$, and weighting is not needed. When $\mu_1$ is near zero and/or $\sigma_2$ is larger than $\sigma_0$, most of the mass of $\Pr(Y|M=1)$ overlaps the mass of $\Pr(Y)$, and all intervals tend to over-cover. This is indicated by the blue regions in the upper-left corner of the coverage plots, and is larger in the weighed methods than the unweighted methods. Conversely, when $\Pr(Y|M=1)$ does adequately overlap the mass of $\Pr(Y)$ because $\mu_1$ is farther from $0$ and/or $\sigma_1$ is less than $\sigma_0$, we see undercoverage indicated by the orange in the lower-right of the plots. Again, this is mitigated by weighting. The non-parametric methods provide better coverage than the normal-theory methods; this is not surprising since $Y$ is not normal in most cases. The width plots verify that the weighted methods bring the extreme widths observed from the unweighted methods closer to optimal.

This example verifies that the weighted methods we propose can substantially reduce over- and under-coverage in cases where there is mismatch between the observed distribution and the distribution of interest. However, they cannot eliminate it entirely when the distributions of $Y$ and $Y|M\is 1$ are very different. This is to be expected; estimating say the mean of one distribution using an importance-weighted sample is challenging in practice. Estimating the {\em tails} of that distribution is even more challenging. Nonetheless, there is value in the weighted approach, and we will explore it further in the DTR setting in the next section.

\subsection{Residual Borrowing}
\newcommand{\Ed}{\mathcal{E}}
We now present a different approach to ensuring that our analysis captures the joint distribution $Y,S_2|S_1,A_1$ correctly, and hence captures variability in $Y|S_1,A_1$ correctly when we marginalize over $S_2$. To do so, we return to the $Q$-learning approach, which estimates $\E[Y|S_2,A_2]$ using regression. As discussed above, the pseudooutcome $\tilde y$ for each trajectory represents our best estimate of $\E[Y|S_1 \is s_1, A_1\is a_1, S_2\is s_2, A_2]$ when $A_2 \sim \pi^*_2(s_1,a_1,s_2)$. This estimate is available for all trajectories in our dataset, including those for which $M\is 1$. Rather than na\"ively constructing tolerance intervals based on the regression of $\tilde{Y}$ on $S_1$ and $A_1$, we create a new pseudooutcome $\check y$ for each point: For trajectories with $m\is 1$, we set $\check y = y$. For trajectories with $m = 0$, we set $\check{y} = \tilde y + \epsilon$, where $\epsilon \sim \Ed$, and $\Ed$ is an estimate of the distribution of the residuals among trajectories with $M = 1$. We call this procedure {\em residual borrowing.} We then construct tolerance intervals using the regression of $\check y$ on $S_1$ and $A_1$.

Unlike the $\tilde y$, the $\check y$ retain information about the distribution of $Y|S_2,A_2$. Furthermore, since we use all of the trajectories in our original dataset, our empirical distribution of $S_2|S_1,A_1$ is representative of the true generative model. The distribution $\Ed$ could be the empirical distribution of the appropriate residuals, or it could be a smoothed estimate, e.g., a kernel density estimate. In our simulations, we found that a smoothed estimate works better than sampling from the empirical distribution.

\section{Empirical Results}\label{ss:empirical}

\begin{table}
\caption{Plot Acronyms for Tolerance Intervals\label{tab:legend}}
\begin{tabular}{l|l}
{\bf RBQNPTI} & Residual-Borrowing Non-parametric TI \\
{\bf RBQTI} & Residual-Borrowing Normal-theory TI \\
{\bf UNPTI} & Unweighted Non-parametric TI \\
{\bf UTI} &  Unweighted Normal-theory TI \\
{\bf WNPTI} & Weighted Non-parametric TI \\
{\bf WTI} & Weighted Normal-theory TI \\
\end{tabular}
\end{table}

We now present results of six tolerance interval methods, which are listed in Table~\ref{tab:legend}, using a simulation study. Our goals are to: 1. verify that inverse probability weighted methods can succeed where the unweighted methods fail, and test their limits; and, 2. to assess the difference in performance between the inverse probability weighted methods and the residual borrowing methods. Note that we do not include results from the na\"ive method as it performs very poorly.

The generative model from the study is taken from \cite{Schulte2014}, with modifications. We begin by reviewing that model and discussing our modifications to it; we then present and discuss the performance of our methods.

\subsection{Generative Model}

The generative model has $2$ decision points. $S_1$ is binary, $A_1$ is binary, $S_2$ is continuous, $A_2$ is binary, and $Y$ is continuous. The generative model under the exploration DTR is given by
\newcommand{\psif}{\xi_{\psi}}
\newcommand{\phif}{\xi_{\phi}}
\newcommand{\Ydist}{\mathrm{Ydist}}
{
\begin{align*}
S_{1} & \sim  \text{Bernouli}(0.5) \nonumber\\
A_{1}^0|S_{1}\is s_{1} & \sim  \text{Bernoulli}\{\text{expit}(\boxed{\phif{}}\{\phi_{10}^{0}+\phi_{11}^{0}s_{1}\})\}\\
S_{2}|S_{1}\is s_{1},A_{1}\is a_{1} & \sim  \text{Normal}(\delta_{10}^{0}+\delta_{11}^{0}s_{1}+\delta_{12}^{0}a_{1}+\delta_{13}^{0}s_{1}a_{1},2)\\
A_{2}^0|S_{1}\is s_{1}, S_{2}\is s_{2}, A_{1}\is a_{1} & \sim  \text{Bernoulli}\{\text{expit}(\boxed{\phif{}}\{\phi_{20}^{0}+\phi_{21}^{0}s_{1}+\phi_{22}^{0}a_{1}
+\phi_{23}^{0}s_{2}+\phi_{24}^{0}a_{1}s_{2}+\phi_{25}^{0}s_{2}^{2}\})\}\\
Y|S_{1}\is s_{1},S_{2}\is s_{2},A_{1}\is a_{1},A_{2}\is a_{2} & \sim  \boxed{\Ydist{}}\{\mu_Y(s_{1},s_{2},a_{1},a_{2}),\boxed{\sigma^2_\varepsilon}\}\\
\mu_Y(s_{1},s_{2},a_{1}a_{2}) & = \beta_{20}^{0}+\beta_{21}^{0}s_{1}+\beta_{22}^{0}a_{1}+\beta_{23}^{0}s_{1}a_{1}
+\beta_{24}^{0}s_{2}+\beta_{25}^{0}s_{2}^{2} 
+a_{2}\boxed{\psif}(\psi_{20}^{0}+\psi_{21}^{0}a_{1}+\psi_{22}^{0}s_{2})
\end{align*}
}

\noindent Here, $\text{expit}(x)=e^{x}/ (e^{x}+1)$.
The original model is indexed by
\begin{equation*}
\begin{array}{llrrrrr}
\phi_{1}^{0} & = (\phi_{10}^{0},& \phi_{10}^{0})\\
& = (0.3,& -0.5)\\[1mm]
\delta_{1}^{0} & = (\delta_{10}^{0},& \delta_{11}^{0},& \delta_{12}^{0},& \delta_{13}^{0}) \\
& = (0,& 0.5,& -0.75,& 0.25)\\[1mm]
\phi_{2}^{0} & =(\phi_{20}^{0},& \phi_{21}^{0},& \phi_{22}^{0},& \phi_{23}^{0},& \phi_{24}^{0},& \phi_{25}^{0}) \\
& = (0,& 0.5,& 0.1,& -1,& -0.1,& 0)\\[1mm]
\beta_{2}^{0} & =(\beta_{20}^{0},& \beta_{21}^{0},& \beta_{22}^{0},& \beta_{23}^{0},& \beta_{24}^{0},& \beta_{25}^{0}) \\
& = (3,& 0,& 0.1,& -0.5,& -0.5,& 0)\\[1mm]
\psi_{2}^{0} & =(\psi_{20}^{0},& \psi_{21}^{0},& \psi_{22}^{0})\\
& = (1,& 0.25,& 0.5)
\end{array}
\end{equation*}
to which we have added four parameters: $\psif$ is a factor multiplying $\psi_2^0$, its default is 1; $\phif$ is a factor multiplying $\phi_2^0$, its default is 1; $\Ydist(\mu,\sigma_\varepsilon^2)$ gives the conditional distribution of $Y$ with given mean and variance; its default is the normal distribution and the default $\sigma^2_\varepsilon$ is 10. We have emphasized these parameters by displaying them in boxes.

	Our parameter $\phif$ allows us to control the degree to which state information influences treatment selection under the exploration (data-gathering) DTR. For $\phif = 1$, we have the original exploration used by Schulte et al., and for $\phif = 0$, we have uniform randomization over treatments independent of state and previous treatment. $\psif$ allows us to control the effect of treatment $A_2$ on $Y$. For $\psif = 1$ we have the treatment effect specified by Schulte et al., and for $\psif = 0$ we have no treatment effect at the second stage. $\Ydist$ allows us to control the shape of the error distribution to see its effect on the tolerance interval methods; Schulte et al.\ used a normal error, but we will explore heavier and lighter-tailed errors while holding variance constant.
	
	This family of generative models allows us to explore what happens to the performance of tolerance interval methods when we have dependence of $S_2$ on $A_2$ during the generating process. While most of the SMART studies we are aware of use a simple randomization strategy where the distribution of $A_2$ does not depend on $S_2$ (which is the case here when e.g.\ $\phif = 0$, giving a simple 50:50 randomization strategy), we expect that more studies akin to ``adaptive trials" with state-dependent randomization will become attractive in the future.

\newcommand{\Ind}{I}

Based on the function\footnote{Denoted $m$ by Schulte et al.} $\mu_Y$ which determines the expected value of $Y|S_1,A_1,S_2,A_2$, we can immediately see that the optimal second stage decision function is 
\begin{align*}
\pi^*_2(a_1,s_2) & = \arg\max_{a_2'} a_2' \psif (\psi_{20}^{0}+ \psi_{21}^{0} a_{1}+\psi_{22}^{0} s_{2}) \\
& = \Ind \{ \psif (\psi_{20}^{0}+ \psi_{21}^{0} a_{1}+\psi_{22}^{0} s_{2}) > 0\}.
\end{align*}

\subsection{Working Model}

Our working model for $Q_2$ is
\begin{align}\label{eq:Q2working}
Q_2(s_1,a_1,s_2,a_2;\beta_2,\psi_2) = 
\beta_{20}+\beta_{21}s_{1}+\beta_{22}a_{1}+\beta_{23}s_{1}a_{1}
+\beta_{24}s_{2}+\beta_{25}s_{2}^{2} 
+ a_{2}(\psi_{20}+\psi_{21}a_{1}+\psi_{22}s_{2})
\end{align}
Having computed least squares estimates $\hat\beta_2$ and $\hat\psi_2$, our estimate of the optimal second-stage decision function is
\begin{equation}\label{eq:pi2working}
\pihat_2^*(s_1,a_1) = \Ind \{\hat\psi_{20}+ \hat\psi_{21} a_{1}+\hat\psi_{22}^{0} s_{2} > 0\}
\end{equation}
and the pseudooutcome for the $i$th trajectory is
\begin{align*}
\tilde{y}_i = \hat\beta_{20}+\hat\beta_{21}s_{1i}+\hat\beta_{22}a_{1i}+\hat\beta_{23}s_{1i}a_{1i}
+\hat\beta_{24}s_{2i}+\hat\beta_{25}s_{2i}^{2} 
+ |\hat\psi_{20}+ \hat\psi_{21} a_{1i}+\hat\psi_{22}^{0} s_{2i}|_+.
\end{align*}

Our working model for $Q_1$ is the saturated model 
\begin{equation}\label{eq:Q1working}
Q_1(s_1,a_1;\beta_1,\psi_1) = \beta_{10}+\beta_{21}s_{1}+a_1(\psi_{10} + \psi_{11}s_1).
\end{equation}
Having computed least squares estimates $\hat\beta_2$ and $\hat\psi_2$ by regressing the pseudooutcomes on $s_1$ and $a_1$, our estimate of the optimal first-stage decision function would be\footnote{\cite{Schulte2014} give the true optimal values of $\beta_1$ and $\psi_{1}$ as a function of the other model parameters.}
\begin{equation}
\pihat_1^*(s_1,a_1) = \Ind \{\hat\psi_{10}+ \hat\psi_{11} a_{1} > 0\}.
\end{equation}

\subsection{Tolerance Intervals}

In many studies of DTR methods, the focus is on point and interval estimates of the optimal stage 1 decision parameters \citep{Chakraborty2010,Chakraborty2013,Chakraborty2013a,laber14dynamic,Chakraborty2014}. In this work, we will investigate methods for constructing tolerance intervals for
\begin{align*}
Y|S_1\is 0,A_1\is 0;\pihat^*_2 ~~~~&~~~~ Y|S_1\is 0,A_1\is 1;\pihat^*_2\\
Y|S_1\is 1,A_1\is 0;\pihat^*_2 ~~~~&~~~~ Y|S_1\is 1,A_1\is 1;\pihat^*_2.
\end{align*}
Note that our goal is to construct tolerance intervals for $Y$ under the {\em estimated optimal regime} rather than under the optimal regime. The reason for this is pragmatic: we assume that it is the estimated optimal regime that would be deployed in future to support decision-making. 

We begin by estimating $\pihat^*_2$ using the working models (\ref{eq:Q2working},\ref{eq:pi2working}). We then compute the pseudooutcome $\tilde y_i$ for each trajectory, and the match indicator $m_i = \Ind\{\pihat^*_2(s_{1i},a_{1i},s_{2i}) = a_{2i}\}$.

\subsubsection{Unweighted Methods}

To construct the unweighted normal-theory TIs, we regress $y$ on $s_1$ and $a_1$ according to working model (\ref{eq:Q1working}) but using only trajectories with $m=1$. We then apply (\ref{eq:normRegTI}) to construct the four tolerance intervals.

To construct the unweighted nonparametric TIs, we divide the trajectories with $m = 1$ into four mutually exclusive groups according to their $(s_1,a_1)$ values. We then construct the four tolerance intervals by applying the Wilks method (\ref{eq:wilksRanks}) to each group.

\subsubsection{Weighted Methods}

To construct the weights, we first form kernel density estimates $\hat f_\Ed(s_2;a_1,a_1,m\is 1)$ for $S_2|S_1\is s_1,A_1\is a_1,M\is 1$ and $\hat f_\Ed(s_2;s_1,a_1)$ for $S_2|S_1\is s_1,A_1\is a_1$. The weight for a trajectory with index $i$ that has $m = 1$ is then given by
\begin{equation}
w_i = \frac{\hat{f}_\Ed (s_{2i};s_{1i},a_{1i})}{\hat{f}_\Ed (s_{2i};s_{1i},a_{1i},m=1)}.
\end{equation}
While logistic regression might be viewed as a more obvious choice for this task, we found that its attendant monotonicity assumptions were often violated, and that the pair of kernel density estimates were the simplest way to produce a more flexible model in this low-dimensional setting.

To construct the weighted normal-theory TIs, as above we compute a weighted regression of $y$ on $s_1$ and $a_1$ according to working model (\ref{eq:Q1working}) but using only trajectories with $m=1$. We then apply (\ref{eq:normRegTI}) to construct the four tolerance intervals; in this case, we use the sandwich estimate \citep{Huber1967,White1980} with the weights to compute $\hat\sigma_{Y|X=x}$. This makes the method somewhat more robust.

To construct the unweighted nonparametric TIs, we divide the trajectories with $m = 1$ into four mutually exclusive groups according to their $(s_1,a_1)$ values. We then construct the four tolerance intervals by applying our weighted modification of the Wilks method (\ref{eq:wilksRanks}) to each group.

\subsubsection{Residual Borrowing}

For the residual borrowing methods, within each $(s_1,a_1)$ group, we first form a kernel density estimate $\hat{f}_R(r;s_1,a_1)$ using the residuals $y_i - \tilde{y}_i$ among the trajectories with $m = 1$.  We then set $\check{y}_i = y_i$ for each trajectory with $m_i=1$, and sample $\check{y}_i$ from the kernel density estimate for trajectories with $m_i=0$. We then either regress $\check{y}_i$ using the working models to create the regression tolerance intervals, or we again divide up the data according to $s_1$ and $a_1$ to construct non-parametric tolerance intervals.

\subsection{Results}

Using the foregoing generative model, working models, and tolerance interval methods, we ran a suite of simulations to investigate performance. Experiments varied by $\phif,\psif,\sigma^2_\varepsilon$, and $\Ydist$, for a total of $1,089$ different experimental settings. Both $\phif$ and $\psif$ were varied from $0$ to $1$ in $0.1$ increments, and $\sigma^2_\varepsilon$ took values in $\{10,1,0.1\}$. We examined settings with $\Ydist$ as normal, uniform, and $t$ with 3 degrees of freedom, each scaled to have the appropriate $\sigma^2_\varepsilon$. For each setting, we drew 1000 simulated datasets each of size $n=1000$, computed tolerance intervals using each of the six methods, and evaluated their {\em content}, that is, what proportion of $Y$ was captured by each interval, and their {\em relative width}, given by $(\TIu - \TIl) / h^*$, where $h*$ is the width of the optimal tolerance interval computed using the $\gamma/2$ and $1 - \gamma/2$ quantiles of the true distribution. For all experiments, we set $1 - \alpha = 0.95$ and $\gamma = 0.9$. All kernel density estimates were one-dimensional, and used the default optimal bandwidth. All experimental code was written in R \citep{R2015}, and is publicly available.

Figures~\ref{fig:10_ALL}, \ref{fig:1_ALL}, and \ref{fig:.1_ALL} display the results of all of our experiments as heatmaps using the same approach as Figure~\ref{fig:meanTIs}. Monte Carlo coverages that are not statistically significantly different from 0.95 are coloured pure white, Over-coverage is coloured blue, and under-coverage is coloured orange. The second row of each subplot gives the average width of the tolerance intervals.

Figure \ref{fig:10_NORM} contains the original model setting proposed by \cite{Schulte2014} in the upper-right corners of its heatmaps. In this setting, the weighted and unweighted normal-theory tolerance intervals undercover slightly, while the weighted and unweighted non-parametric methods overcover, and are much wider. The residual-borrowing methods perform best in this setting, with the normal-theory residual-borrowing intervals achieving near-nominal coverage with modest width. There is relatively little variation in coverage and width across $\phif$ and $\psif$ in this setting, we believe because the noise level is quite high relative to the effect of $a_2$ even when $\psif = 1$. In Figure \ref{fig:10_T3}, $\Ydist$ was chosen to be a $t$-distribution with 3 degrees of freedom, scaled to have variance $\sigma^2_\varepsilon = 10$ and shifted by $\mu_Y$. In this heavy-tailed setting, it is the non-parametric residual-borrowing method that slightly undercovers, while the other methods overcover somewhat. As in the normal case, the weighted and unweighted nonparametric methods are very wide. Figure \ref{fig:10_UNIF} uses a scaled and shifted uniform distribution for $\Ydist$, again maintaining $\sigma^2_\varepsilon = 10$. In this light-tailed setting, in contrast to Figure~\ref{fig:10_T3}, it is the normal-theory intervals which tend to be wide, while the non-parametric ones are narrower. The residual-borrowing intervals are wide as well. All intervals achieve nominal or greater coverage in this setting.

\begin{figure}
\begin{subfigure}{\columnwidth}
\includegraphics[width=\columnwidth]{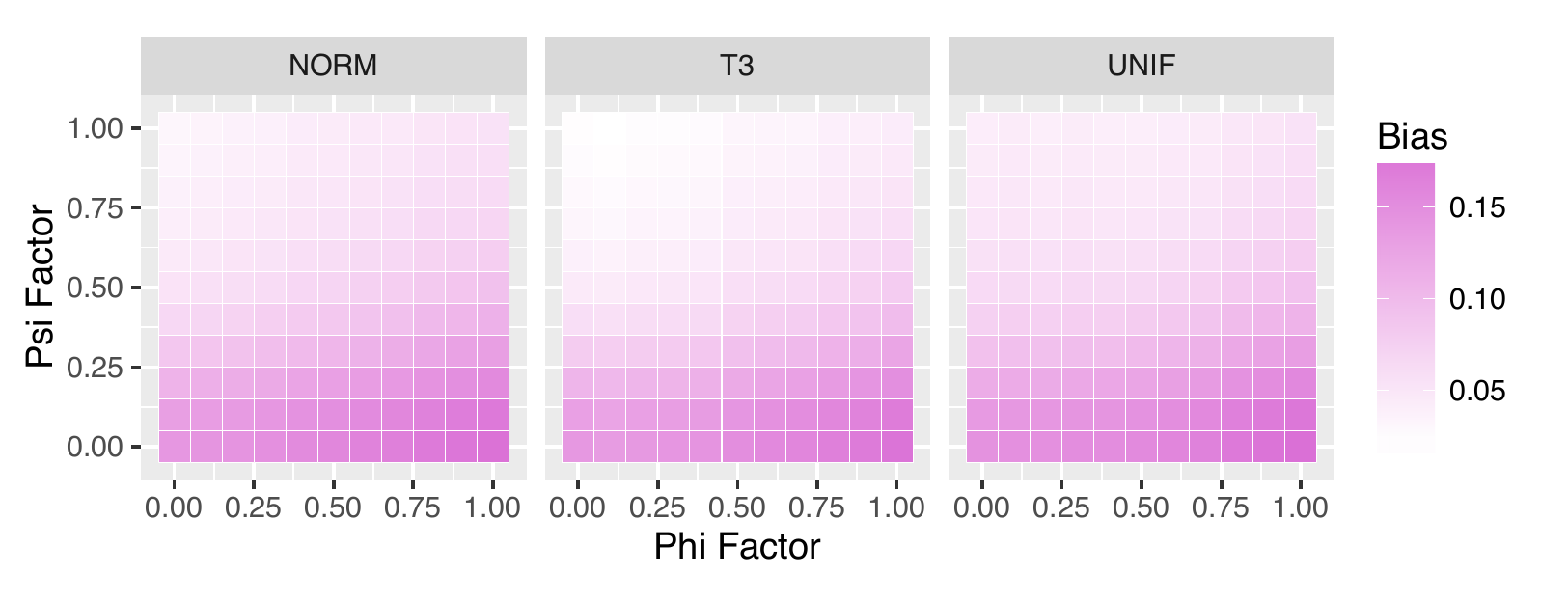}
\caption{$\sigma^2_\varepsilon ={10}$ ~~~~~~~~\\ ~\\}
\end{subfigure}
\begin{subfigure}{\columnwidth}
\includegraphics[width=\columnwidth]{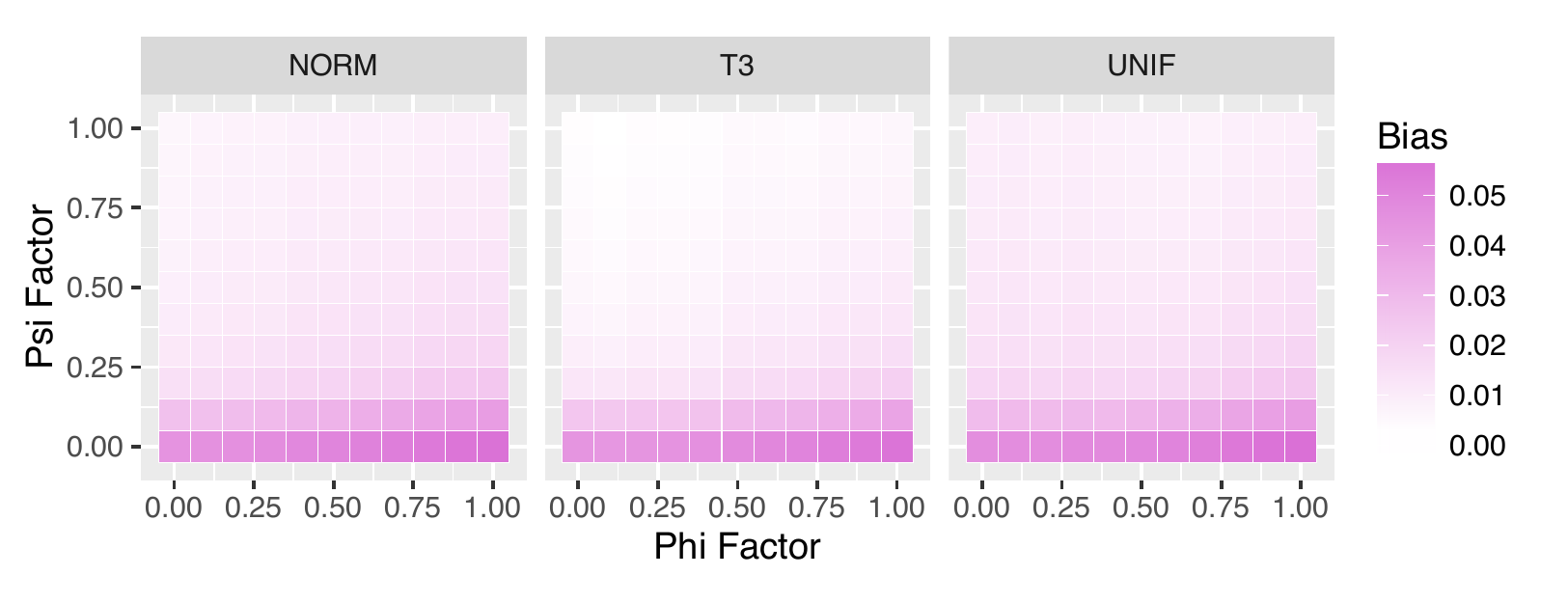}
\caption{$\sigma^2_\varepsilon=1$ ~~~~~~~~\\ ~\\}
\end{subfigure}
\begin{subfigure}{\columnwidth}
\includegraphics[width=\columnwidth]{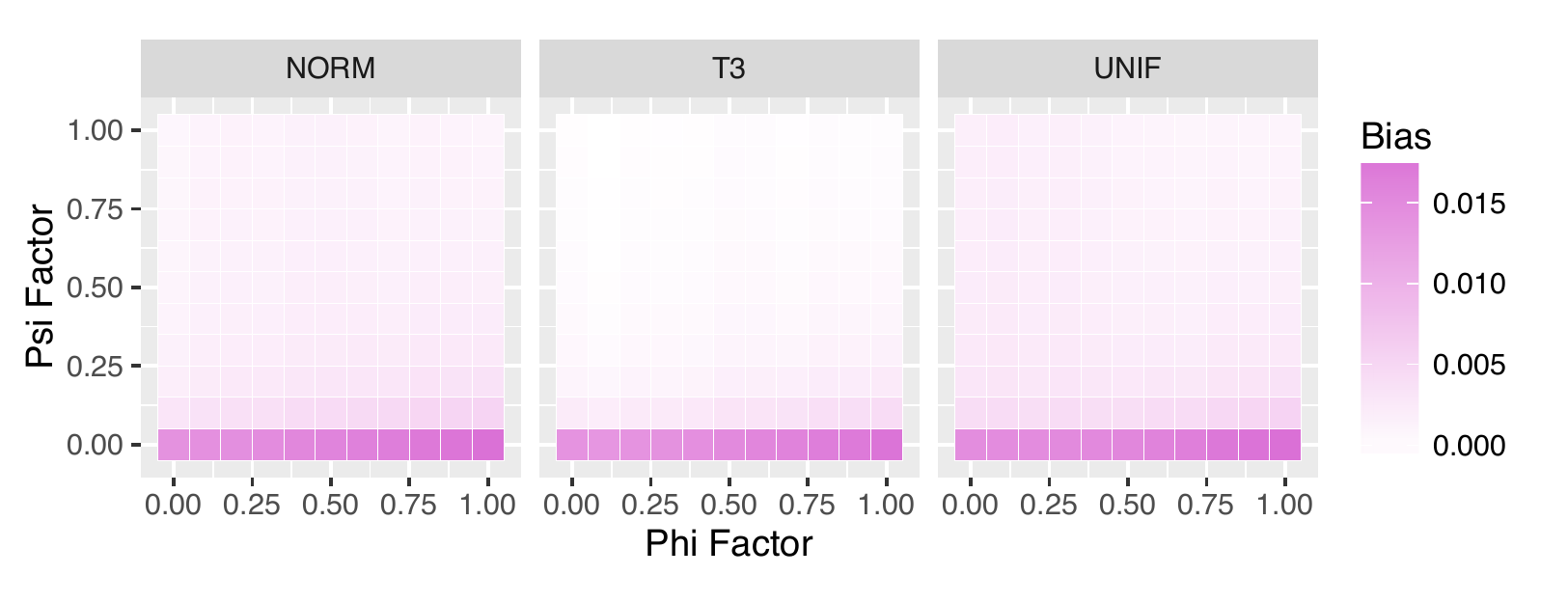}
\caption{$\sigma^2_\varepsilon=0.1$ ~~~~~~~~\\ ~\\}
\end{subfigure}
\caption{\label{fig:BIAS}Bias in the estimated value of the optimal policy. This is a surrogate measure of non-regularity; note that maximal bias occurs when $\psi$, which controls the effect of $A_2$, is small. Phi Factor ($\phif$ in the text) controls the effect of covariates on the exploration DTR, and Psi Factor ($\psif$ in the text) controls the effect of $A_2$ on the conditional mean of $Y$.}
\end{figure}

We see a striking change as we examine the lower-noise settings in Figure~\ref{fig:1_ALL}, which have $\sigma^2_\varepsilon = 1$. Here, we start to see dependence of performance on $\psif$ and $\phif$. As in Figure~\ref{fig:10_NORM}, in Figure~\ref{fig:1_NORM} we see the normal theory intervals undercovering, although we now see a definite trend that worsens as $\phif$ increases, and as $\psif$ decreases. We also see this trend among the non-parametric methods, which range from overcovering to undercovering as we move across $\phif$ and $\psif$. Overall, we see the greatest coverage when the effect of $A_2$ is quite strong (topmost rows), or if the dependence of $A_2^0$ on $S_2$ is weak (leftmost rows.) As we discussed earlier, when $\phif = 0$ (leftmost columns) there is no dependence of $M$ on $S2$, and thus weighting is unnecessary. Furthermore, we not only obtain a uniform probability of $M=1$ across $S_2$, but also a uniform probability of $A_2^0$ across $S_2$. This uniformity likely leads to improved estimates of $Y|S_2,A_2$, and in turn to better coverage of the tolerance intervals. The decrease in performance for low $\psif$ may be due to non-regularity: when $\psif=0$, there is in fact no effect of $A_2$ on $Y$. However, assuming continuity of the appropriate distributions, our estimated $\hat\psi_2$ will be nonzero almost always, and our plug-in estimate of the value of $\pihat_2^*$ will be positive almost always. Defining $\hat{a}_{2i}^* = \pihat_2^*(s_{1i},a_{1i},s_{2i})$, the empirical bias in the value of $\pihat_2^*$ is 
\begin{multline*}
\sum_i \hat{a}_{2i}^*(\hat\psi_{20} + \hat\psi_{21} a_{1i}+\hat\psi_{22} s_{2i}) - \\
\hat{a}_{2i}^*\psif (\psi_{20}^0 + \psi_{21}^0 a_{1i}+ \psi_{22}^{0} s_{2i}).
\end{multline*}
Figure~\ref{fig:BIAS} shows the average empirical bias in our estimate of the average value of using $\pihat_2^*$, as a function of $\phif$ and $\psif$. We can see that the bias is concentrated at the bottom of the plots, near $\psif = 0$. This is precisely where there is more than one nearly-optimal action and non-regularity is known to be a problem. 

We see the problems worsen in Figure~\ref{fig:.1_ALL}, where we set $\sigma^2_\epsilon = 0.1$. We hypothesise that this is because proportionately even more of the variability in $Y$ is attributable to variability in $S_2$, and accurate estimation of $Y|S_2,A_2$ becomes that much more important. All of the matched subset methods have severe undercoverage for large values of $\phi$ and low values of $\psi$. Weighted methods mitigate this. The residual-borrowing methods achieve much better coverage, but at the cost of much wider intervals.

\subsection{Discussion}\label{ss:discussion}

Based on our simulation study experiments, we believe that designing the exploration DTR to have uniform randomization over actions is highly beneficial for estimating tolerance intervals. When this is the case, all methods gave reasonable results in almost all scenarios. Some knowledge of the error distribution may help choose a method that will result in reasonable widths. If uniform exploration is not possible, the residual-borrowing methods appear to be the most robust to undercoverage, followed by the weighted methods, followed by the unweighted methods. That said, it would be prudent to perform a simulation study under a scenario ``close'' to the analysis at hand if possible; to facilitate this we have released our R code \citep{R2015} so that researchers and practitioners can explore other scenarios.

\section{Example: STAR*D}

We present an example of the application of the TI methods we have described to real-world clinical trial data. The Sequenced Treatment Alternatives to Relieve Depression (STAR*D) study followed an initial population of 4041 patients as they were treated using different antidepressant medications and cognitive behavioural therapy \citep{Rush2004}. There were a total of three decision points at which randomisation took place, with different treatment options available at each one. Outcomes were measured using the clinician-rated Quick Inventory of Depressive Symptomatology \citep{Rush2003}. We will examine two such decision points corresponding to Level 2 and Level 3 of the study, which will correspond to the first and second decision points in our analysis.

We construct tolerance intervals for STAR*D at Level 2 (our decision point 1), having estimated a $Q$ function and estimated optimal policy for Level 3 (our decision point 2.) We use exactly the same $Q$-learning working model and estimation procedure as \cite{Schulte2014} to develop $\pihat_2^*$ and the pseudooutcomes; we refer the interested reader to their work for more details. In summary, the state variables we use are up-to-date QIDS measures of patient symptoms, and the outcomes we use are based on later QIDS measurements that have been negated so that higher values are preferable. At decision point 1, we elect to use a binary state variable indicating whether the previous slope in QIDS score for a patient is greater than the median. Higher QIDS scores indicate worse symptom levels, so this state variable effectively identifies patients whose disease status is worsening most quickly. At both decision points, the treatment choice is whether to ``augment'' the current medication with another, or to ``switch'' to another medication altogether.

We applied the six TI methods described previously to the data, using the choice to switch or augment treatment as $A_1$, and letting $S_1$ be an indicator variable for QIDS slope being greater than the median slope. We see that generally the intervals are quite wide, and that there is severe overlap of TIs for different treatments. This reflects the high variance and low treatment effect we observe in this data. However, the intervals do capture prognostic information: the intervals for $S_1=\mbox{``Yes''}$ (indicating severely worsening symptoms) are wider, with a decreased lower bound indicating that such patients may have poorer outcomes relative to those with more stable symptoms prior to the decision point. The maximum attainable outcome in this problem is $0$, since QIDS cannot go below $0$. We note that the parametric TI methods can produce upper bounds greater than $0$ and lower bounds that appear to be a bit optimistic. Hence, we suggest that one of the non-parametric methods would be a sensible choice for STAR*D.

\begin{figure*}
\includegraphics[width=\textwidth]{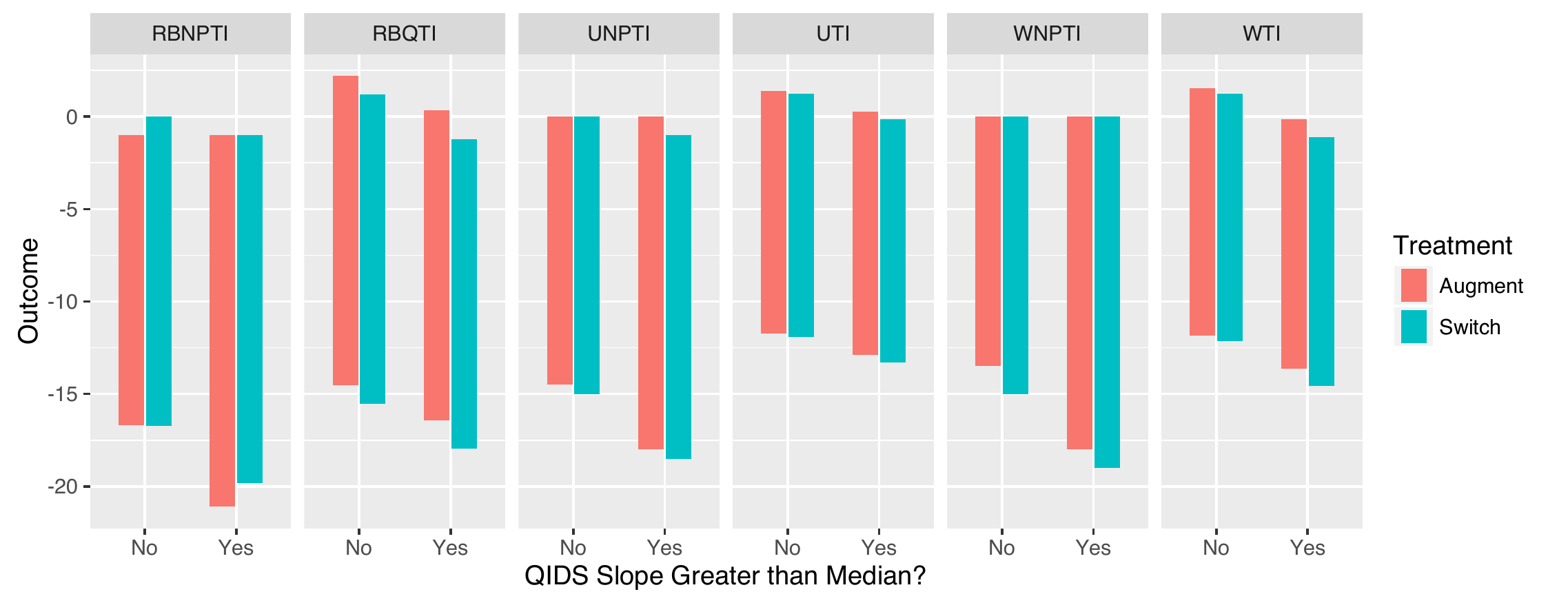}
\caption{\label{fig:stard}Tolerance intervals for STAR*D at Level 2. The six TI methods used previously are applied to the data, using the choice to switch or augment treatment as $A_1$, and letting $S_1$ be an indicator variable for previous QIDS slope being greater than the median. In this setting, higher outcomes are preferable, but higher QIDS scores (and slopes) indicate worse symptoms.}
\end{figure*}

\section{Conclusion}\label{ss:conclusion}

We have developed and evaluated tolerance interval methods for dynamic treatment regimes that can provide more detailed prognostic information to patients who will follow an estimated optimal regime. We began by reviewing in detail different interval estimation and prediction methods and then adapting them to the DTR setting. We illustrated some of the challenges associated with tolerance interval estimation stemming from the fact that we do not typically have data that were generated from the estimated optimal regime. We gave an extensive empirical evaluation of the methods and discussed several practical aspects of method choice. We demonstrated the methods using data from a pragmatic clinical trial. We now take the opportunity to discuss future directions of research on tolerance intervals for dynamic treatment regimes.

\subsection{Future Directions}
Our work lays the foundation for extending tolerance interval methods for dynamic treatment regimes in several different directions.

The normal theory TI methods we employed used an estimate of the residual distribution that is pooled over $S_1$ and $A_1$. The non-parametric methods estimated the residual distributions separately for the different discrete $S_1,A_1$. A compromise solution that partially shares residual information across different configurations of $(S_1,A_1)$, perhaps in a data-driven, adaptive fashion, may provide improved performance and wider applicability. (Note that the non-parametric methods we described are not applicable if $S_1$ is continuous.)

We have treated DTRs with two decision points, but in general we would like to have tolerance intervals for multiple decision points. Such methods would potentially have to address uncertainty stemming from ``parameter sharing,'' across time points. It is known \citep{Chakraborty2016} that the effects of model misspecification and non-regularity can compound in the multiple decision point setting, and the impact of this on tolerance intervals is not yet known.

While we assumed a single outcome measure $Y$ throughout our work, several methods have been described for estimating DTRs in the presence of multiple outcomes \citep{lizotte12linear,laber14setvalued,lizotte15momdps}. Joint tolerance intervals/tolerance regions for this setting would be equally important as they are in the standard, single-outcome setting.

We observed some problems associated with biased estimates of the value of the estimated policy, which is caused by non-regularity. The problem of non-regularity in optimal DTR estimation has been addressed in the confidence interval setting using different approaches, including pre-testing \citep{laber14dynamic} and shrinkage \citep{Chakraborty2010,Chakraborty2013}. We have not explicitly incorporated either of these ideas in the methods we presented; doing so may lead to methods that are more robust to small or zero treatment effects at the second stage yet do not pay a high cost in terms of width.

\cite{Fernholz2001} have presented a method to re-calibrate tolerance intervals using the bootstrap. They propose a bootstrap method to estimate the content $\gamma$ of a given tolerance interval---first they construct a tolerance interval with nominal (or ``requested'') content $\gamma$, but then they use the bootstrap to estimate what the actual content. This could potentially be used to identify when tolerance methods fail on dynamic treatment regimes, or they may be used simply to give more accurate confidence information to the decision maker. For example, we may attempt to construct a tolerance interval for $\gamma = 0.9$, but if it turns out that the actual content is $0.85$, the interval may still be useful if the decision-maker is made aware of this fact. Future work to adapt the calibration procedure could prove promising.

Finally, a Bayesian approach to the predictive estimation problem may prove fruitful in some settings. \cite{Saarela2015} have laid groundwork for this direction of research.


\subsection*{Acknowledgements}
This work was supported by the Natural Sciences and Engineering Research Council of Canada. Data used in the preparation of this article were obtained from the limited access datasets distributed from the NIH-supported ``Sequenced Treatment Alternatives to Relieve Depression'' (STAR*D). STAR*D focused on non-psychotic major depressive disorder in adults seen in outpatient settings. The primary purpose of this research study was to determine which treatments work best if the first treatment with medication does not produce an acceptable response. The study was supported by NIMH Contract \#N01MH90003 to the University of Texas Southwestern Medical Center. The ClinicalTrials.gov identifier is NCT00021528.

\bibliographystyle{plainnat}

\bibliography{predint}

\begin{thebibliography}{39}
\providecommand{\natexlab}[1]{#1}
\providecommand{\url}[1]{\texttt{#1}}
\expandafter\ifx\csname urlstyle\endcsname\relax
  \providecommand{\doi}[1]{doi: #1}\else
  \providecommand{\doi}{doi: \begingroup \urlstyle{rm}\Url}\fi

\bibitem[Barrett et~al.(2014)Barrett, Henderson, and Rosth{\o}j]{Barrett2014}
Jessica~K. Barrett, Robin Henderson, and Susanne Rosth{\o}j.
\newblock {Doubly Robust Estimation of Optimal Dynamic Treatment Regimes}.
\newblock \emph{Statistics in Biosciences}, 6\penalty0 (2):\penalty0 244--260,
  nov 2014.
\newblock ISSN 1867-1764.
\newblock \doi{10.1007/s12561-013-9097-6}.

\bibitem[Barry and Edgman-Levitan(2012)]{Barry2012}
Michael~J. Barry and Susan Edgman-Levitan.
\newblock Shared decision making â the pinnacle of patient-centered care.
\newblock \emph{New England Journal of Medicine}, 366\penalty0 (9):\penalty0
  780--781, 2012.
\newblock \doi{10.1056/NEJMp1109283}.

\bibitem[Blatt et~al.(2004)Blatt, Murphy, and Zhu]{Blatt2004}
D.~Blatt, S.A. Murphy, and J.~Zhu.
\newblock A-learning for approximate planning.
\newblock Technical Report 04-63, The Methodology Center, The Pennsylvania
  State University, University Park, PA, 2004.

\bibitem[Chakraborty and Moodie(2013)]{Chakraborty2013a}
Bibhas Chakraborty and Erica~E.M. Moodie.
\newblock \emph{{Statistical Methods for Dynamic Treatment Regimes}}.
\newblock Statistics for Biology and Health. Springer New York, New York, NY,
  2013.
\newblock ISBN 978-1-4614-7427-2.
\newblock \doi{10.1007/978-1-4614-7428-9}.

\bibitem[Chakraborty et~al.(2010)Chakraborty, Murphy, and
  Strecher]{Chakraborty2010}
Bibhas Chakraborty, Susan Murphy, and Victor Strecher.
\newblock {Inference for non-regular parameters in optimal dynamic treatment
  regimes}.
\newblock \emph{Statistical Methods in Medical Research}, 19\penalty0
  (3):\penalty0 317--343, jun 2010.
\newblock ISSN 0962-2802.
\newblock \doi{10.1177/0962280209105013}.

\bibitem[Chakraborty et~al.(2013)Chakraborty, Laber, and Zhao]{Chakraborty2013}
Bibhas Chakraborty, Eric~B. Laber, and Yingqi Zhao.
\newblock {Inference for Optimal Dynamic Treatment Regimes Using an Adaptive m
  -Out-of- n Bootstrap Scheme}.
\newblock \emph{Biometrics}, 69\penalty0 (3):\penalty0 714--723, sep 2013.
\newblock ISSN 0006341X.
\newblock \doi{10.1111/biom.12052}.

\bibitem[Chakraborty et~al.(2014)Chakraborty, Laber, and Zhao]{Chakraborty2014}
Bibhas Chakraborty, Eric~B Laber, and Y.-Q. Zhao.
\newblock {Inference about the expected performance of a data-driven dynamic
  treatment regime}.
\newblock \emph{Clinical Trials}, 11\penalty0 (4):\penalty0 408--417, aug 2014.
\newblock ISSN 1740-7745.
\newblock \doi{10.1177/1740774514537727}.

\bibitem[Chakraborty et~al.(2016)Chakraborty, Ghosh, Moodie, and
  Rush]{Chakraborty2016}
Bibhas Chakraborty, Palash Ghosh, Erica E.~M. Moodie, and A.~John Rush.
\newblock Estimating optimal shared-parameter dynamic regimens with application
  to a multistage depression clinical trial.
\newblock \emph{Biometrics}, page Epub ahead of print, 2016.
\newblock ISSN 1541-0420.
\newblock \doi{10.1111/biom.12493}.

\bibitem[Collins et~al.(2014)Collins, Nahum-Shani, and Almirall]{Collins2014}
Linda~M Collins, Inbal Nahum-Shani, and Daniel Almirall.
\newblock {Optimization of behavioral dynamic treatment regimens based on the
  sequential, multiple assignment, randomized trial (SMART).}
\newblock \emph{Clinical trials (London, England)}, 11\penalty0 (4):\penalty0
  426--434, jun 2014.
\newblock ISSN 1740-7753.
\newblock \doi{10.1177/1740774514536795}.

\bibitem[Fernholz and Gillespie(2001)]{Fernholz2001}
Luisa~T Fernholz and John~A Gillespie.
\newblock {Content-Corrected Tolerance Limits Based on the Bootstrap}.
\newblock \emph{Technometrics}, 43\penalty0 (2):\penalty0 147--155, may 2001.
\newblock ISSN 0040-1706.
\newblock \doi{10.1198/004017001750386260}.

\bibitem[Ghosh(2011)]{Ghosh2011}
Debashis Ghosh.
\newblock Propensity score modelling in observational studies using dimension
  reduction methods.
\newblock \emph{Statistics \& Probability Letters}, 81\penalty0 (7):\penalty0
  813 -- 820, 2011.
\newblock ISSN 0167-7152.
\newblock \doi{http://dx.doi.org/10.1016/j.spl.2011.03.002}.
\newblock URL
  \url{http://www.sciencedirect.com/science/article/pii/S016771521100085X}.
\newblock Statistics in Biological and Medical Sciences.

\bibitem[Harrell et~al.(2015)Harrell, with contributions~from Charles~Dupont,
  and many others.]{Harrell2015}
Frank~E. Harrell, Jr., with contributions~from Charles~Dupont, and many others.
\newblock \emph{Hmisc: Harrell Miscellaneous}, 2015.
\newblock URL \url{https://CRAN.R-project.org/package=Hmisc}.
\newblock R package version 3.17-1.

\bibitem[Huang et~al.(2015)Huang, Choi, Wang, and Thall]{Huang2015}
Xuelin Huang, Sangbum Choi, Lu~Wang, and Peter~F. Thall.
\newblock {Optimization of multi-stage dynamic treatment regimes utilizing
  accumulated data}.
\newblock \emph{Statistics in Medicine}, 34\penalty0 (26):\penalty0 3424--3443,
  nov 2015.
\newblock ISSN 02776715.
\newblock \doi{10.1002/sim.6558}.

\bibitem[Huber(1967)]{Huber1967}
Peter~J. Huber.
\newblock The behavior of maximum likelihood estimates under nonstandard
  conditions.
\newblock In \emph{Proceedings of the Fifth Berkeley Symposium on Mathematical
  Statistics and Probability, Volume 1: Statistics}, pages 221--233, Berkeley,
  Calif., 1967. University of California Press.

\bibitem[Koller and Friedman(2009)]{Koller2009}
D.~Koller and N.~Friedman.
\newblock \emph{Probabilistic Graphical Models: Principles and Techniques}.
\newblock Adaptive computation and machine learning. MIT Press, 2009.
\newblock ISBN 9780262013192.

\bibitem[{Krishnamoorthy, Kalimuthu and
  Mathew}(2009)]{KrishnamoorthyKalimuthuandMathew2009}
Thomas {Krishnamoorthy, Kalimuthu and Mathew}.
\newblock \emph{{Statistical tolerance regions: theory, applications, and
  computation}}.
\newblock John Wiley \& Sons, 2009.

\bibitem[Laber et~al.(2014{\natexlab{a}})Laber, Lizotte, and
  Ferguson]{laber14setvalued}
Eric~B Laber, Daniel~J Lizotte, and Bradley Ferguson.
\newblock {Set-valued dynamic treatment regimes for competing outcomes}.
\newblock \emph{Biometrics}, 70\penalty0 (1):\penalty0 53--61, mar
  2014{\natexlab{a}}.

\bibitem[Laber et~al.(2014{\natexlab{b}})Laber, Lizotte, Qian, and
  Murphy]{laber14dynamic}
Eric~B Laber, Daniel~J Lizotte, Min Qian, and Susan~A Murphy.
\newblock {Dynamic treatment regimes: technical challenges and applications}.
\newblock \emph{Electronic Journal of Statistics}, 8\penalty0 (0):\penalty0
  1225--1272, 2014{\natexlab{b}}.

\bibitem[Lizotte and Laber(2015)]{lizotte15momdps}
Daniel~J. Lizotte and Eric~B. Laber.
\newblock Multi-objective {M}arkov decision processes for decision support.
\newblock \emph{In submission}, 2015.

\bibitem[Lizotte et~al.(2010)Lizotte, Bowling, and Murphy]{lizotte10multiple}
Daniel~J Lizotte, Michael Bowling, and Susan~A Murphy.
\newblock {Efficient Reinforcement Learning with Multiple Reward Functions for
  Randomized Clinical Trial Analysis.}
\newblock In \emph{Proceedings of the 27th International Conference on Machine
  Learning}, pages 695--702, 2010.

\bibitem[Lizotte et~al.(2012)Lizotte, Bowling, and Murphy]{lizotte12linear}
Daniel~J Lizotte, Michael Bowling, and Susan~A Murphy.
\newblock {Linear Fitted-Q Iteration with Multiple Reward Functions}.
\newblock \emph{Journal of Machine Learning Research}, 13:\penalty0 3253--3295,
  nov 2012.

\bibitem[Moodie(2009)]{Moodie2009}
E.~E~M Moodie.
\newblock {A note on the variance of doubly-robust G-estimators}.
\newblock \emph{Biometrika}, 96\penalty0 (4):\penalty0 998--1004, dec 2009.
\newblock ISSN 0006-3444.
\newblock \doi{10.1093/biomet/asp043}.

\bibitem[Nahum-Shani et~al.(2012{\natexlab{a}})Nahum-Shani, Qian, Almirall,
  Pelham, Gnagy, Fabiano, Waxmonsky, Yu, and Murphy]{Nahum-Shani2012}
Inbal Nahum-Shani, Min Qian, Daniel Almirall, William~E Pelham, Beth Gnagy,
  Gregory~A Fabiano, James~G Waxmonsky, Jihnhee Yu, and Susan~A Murphy.
\newblock {Q-learning: a data analysis method for constructing adaptive
  interventions.}
\newblock \emph{Psychological methods}, 17\penalty0 (4):\penalty0 478--94, dec
  2012{\natexlab{a}}.
\newblock ISSN 1939-1463.
\newblock \doi{10.1037/a0029373}.

\bibitem[Nahum-Shani et~al.(2012{\natexlab{b}})Nahum-Shani, Qian, Almirall,
  Pelham, Gnagy, Fabiano, Waxmonsky, Yu, and Murphy]{Nahum-Shani2012a}
Inbal Nahum-Shani, Min Qian, Daniel Almirall, William~E Pelham, Beth Gnagy,
  Gregory~A Fabiano, James~G Waxmonsky, Jihnhee Yu, and Susan~A Murphy.
\newblock {Experimental design and primary data analysis methods for comparing
  adaptive interventions.}
\newblock \emph{Psychological methods}, 17\penalty0 (4):\penalty0 457--77, dec
  2012{\natexlab{b}}.
\newblock ISSN 1939-1463.

\bibitem[{Neter, John and Wasserman, William and
  Kutner}(1989)]{NeterJohnandWassermanWilliamandKutner1989}
Michael~H {Neter, John and Wasserman, William and Kutner}.
\newblock \emph{{Applied linear regression models}}.
\newblock Irwin Homewood, IL, 1989.

\bibitem[Orellana et~al.(2010)Orellana, Rotnitzky, and Robins]{Orellana2010}
Liliana Orellana, Andrea Rotnitzky, and James~M Robins.
\newblock {Dynamic Regime Marginal Structural Mean Models for Estimation of
  Optimal Dynamic Treatment Regimes, Part II: Proofs of Results}.
\newblock \emph{The International Journal of Biostatistics}, 6\penalty0
  (2):\penalty0 Article9, jan 2010.
\newblock ISSN 1557-4679.
\newblock \doi{10.2202/1557-4679.1242}.

\bibitem[{R Core Team}(2015)]{R2015}
{R Core Team}.
\newblock \emph{R: A Language and Environment for Statistical Computing}.
\newblock R Foundation for Statistical Computing, Vienna, Austria, 2015.
\newblock URL \url{https://www.R-project.org/}.

\bibitem[Rush et~al.(2004)Rush, Fava, Wisniewski, Lavori, Trivedi, Sackeim, and
  et~al.]{Rush2004}
A~J Rush, M~Fava, S~R Wisniewski, P~W Lavori, M~Trivedi, H~A Sackeim, and
  et~al.
\newblock Sequenced treatment alternatives to relieve depression
  ({S}{T}{A}{R}*{D}): rationale and design.
\newblock \emph{Controlled Clinical Trials}, 25\penalty0 (1):\penalty0 119--42,
  Feb 2004.

\bibitem[Rush et~al.(2003)Rush, Trivedi, Ibrahim, Carmody, Arnow, Klein,
  Markowitz, Ninan, Kornstein, Manber, Thase, Kocsis, and Keller]{Rush2003}
A.John Rush, Madhukar~H Trivedi, Hicham~M Ibrahim, Thomas~J Carmody, Bruce
  Arnow, Daniel~N Klein, John~C Markowitz, Philip~T Ninan, Susan Kornstein,
  Rachel Manber, Michael~E Thase, James~H Kocsis, and Martin~B Keller.
\newblock {The 16-Item quick inventory of depressive symptomatology (QIDS),
  clinician rating (QIDS-C), and self-report (QIDS-SR): a psychometric
  evaluation in patients with chronic major depression}.
\newblock \emph{Biological Psychiatry}, 54\penalty0 (5):\penalty0 573--583, sep
  2003.
\newblock ISSN 00063223.
\newblock \doi{10.1016/S0006-3223(02)01866-8}.

\bibitem[Saarela et~al.(2015)Saarela, Arjas, Stephens, and Moodie]{Saarela2015}
Olli Saarela, Elja Arjas, David~A. Stephens, and Erica E~M Moodie.
\newblock {Predictive Bayesian inference and dynamic treatment regimes}.
\newblock \emph{Biometrical Journal}, 57\penalty0 (6):\penalty0 941--958, nov
  2015.
\newblock ISSN 03233847.
\newblock \doi{10.1002/bimj.201400153}.

\bibitem[Schulte et~al.(2014)Schulte, Tsiatis, Laber, and
  Davidian]{Schulte2014}
Phillip~J. Schulte, Anastasios~A. Tsiatis, Eric~B. Laber, and Marie Davidian.
\newblock $q$- and $a$-learning methods for estimating optimal dynamic
  treatment regimes.
\newblock \emph{Statistical Science}, 29\penalty0 (4):\penalty0 640--661, nov
  2014.
\newblock ISSN 0883-4237.

\bibitem[Shortreed et~al.(2011)Shortreed, Laber, Lizotte, Stroup, Pineau, and
  Murphy]{shortreed11mlj}
Susan~M Shortreed, Eric~B Laber, Daniel~J Lizotte, T~Scott Stroup, Joelle
  Pineau, and Susan~A Murphy.
\newblock {Informing sequential clinical decision-making through reinforcement
  learning: an empirical study}.
\newblock \emph{Machine Learning}, 84\penalty0 (1):\penalty0 109--136, 2011.
\newblock ISSN 0885-6125.

\bibitem[Vardeman(1992)]{Vardeman2012}
Stephen~B. Vardeman.
\newblock {What about the other Intervals?}
\newblock \emph{The American Statistician}, 46\penalty0 (3):\penalty0 193--197,
  Feb 1992.

\bibitem[Wallis(1951)]{Wallis1951}
W.~Allen Wallis.
\newblock Tolerance intervals for linear regression.
\newblock In \emph{Proceedings of the Second Berkeley Symposium on Mathematical
  Statistics and Probability}, pages 43--51, Berkeley, Calif., 1951. University
  of California Press.

\bibitem[White(1980)]{White1980}
Halbert White.
\newblock A heteroskedasticity-consistent covariance matrix estimator and a
  direct test for heteroskedasticity.
\newblock \emph{Econometrica}, 48\penalty0 (4):\penalty0 817--838, 1980.
\newblock ISSN 00129682, 14680262.

\bibitem[Wilks(1941)]{Wilks1941}
S.~S. Wilks.
\newblock {Determination of Sample Sizes for Setting Tolerance Limits}.
\newblock \emph{The Annals of Mathematical Statistics}, 12\penalty0
  (1):\penalty0 91--96, mar 1941.
\newblock ISSN 0003-4851.
\newblock \doi{10.1214/aoms/1177731788}.

\bibitem[Young(2013)]{Young2013}
Derek~S. Young.
\newblock {Regression Tolerance Intervals}.
\newblock \emph{Communications in Statistics - Simulation and Computation},
  42\penalty0 (9):\penalty0 2040--2055, oct 2013.
\newblock ISSN 0361-0918.
\newblock \doi{10.1080/03610918.2012.689064}.

\bibitem[Zhao and Laber(2014)]{Zhao2014}
Y.-Q. Zhao and Eric~B Laber.
\newblock {Estimation of optimal dynamic treatment regimes}.
\newblock \emph{Clinical Trials}, 11\penalty0 (4):\penalty0 400--407, aug 2014.
\newblock ISSN 1740-7745.
\newblock \doi{10.1177/1740774514532570}.

\bibitem[Zhao et~al.(2015)Zhao, Zeng, Laber, and Kosorok]{Zhao2015}
Ying-Qi Zhao, Donglin Zeng, Eric~B Laber, and Michael~R Kosorok.
\newblock {New Statistical Learning Methods for Estimating Optimal Dynamic
  Treatment Regimes}.
\newblock \emph{Journal of the American Statistical Association}, 110\penalty0
  (510):\penalty0 583--598, apr 2015.
\newblock ISSN 0162-1459.
\newblock \doi{10.1080/01621459.2014.937488}.

\end{thebibliography}

\begin{figure*}
\begin{subfigure}{\textwidth}
\includegraphics[width=\textwidth]{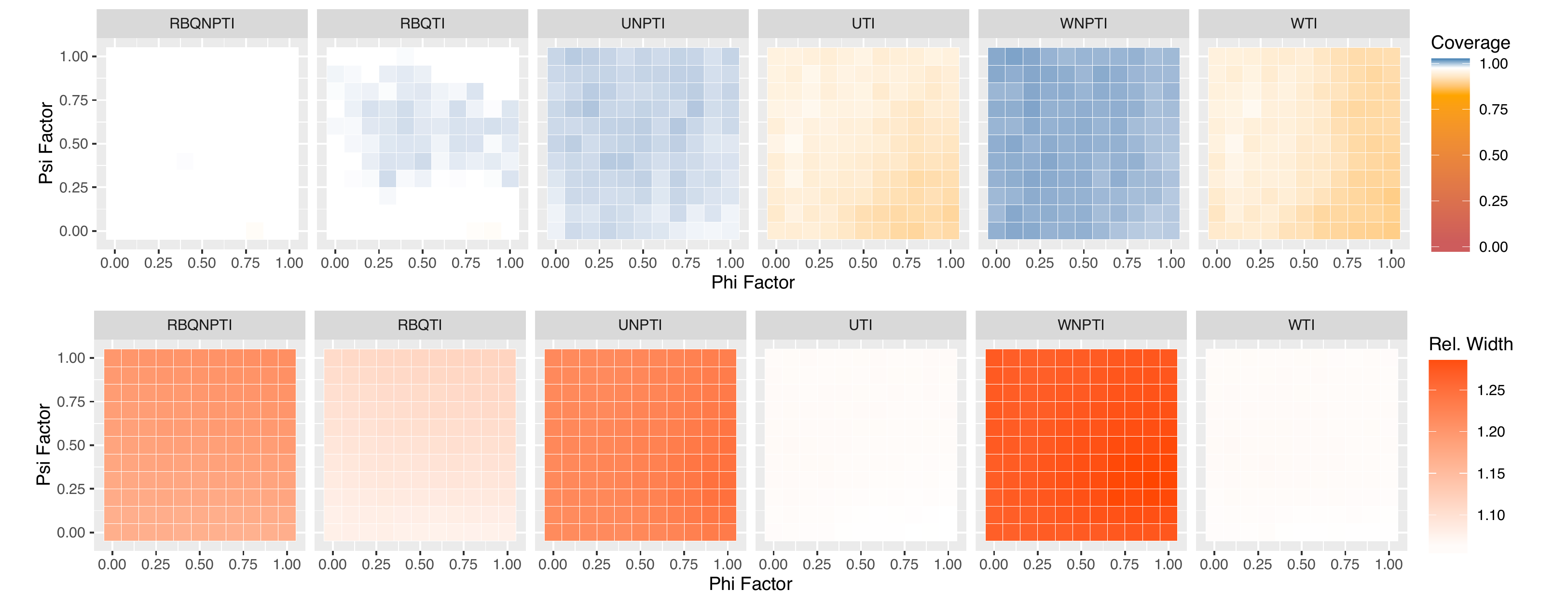}
\caption{\label{fig:10_NORM}{\bf normal} errors ~~~~~~~~~~~\\ ~\\}
\end{subfigure}
\begin{subfigure}{\textwidth}
\includegraphics[width=\textwidth]{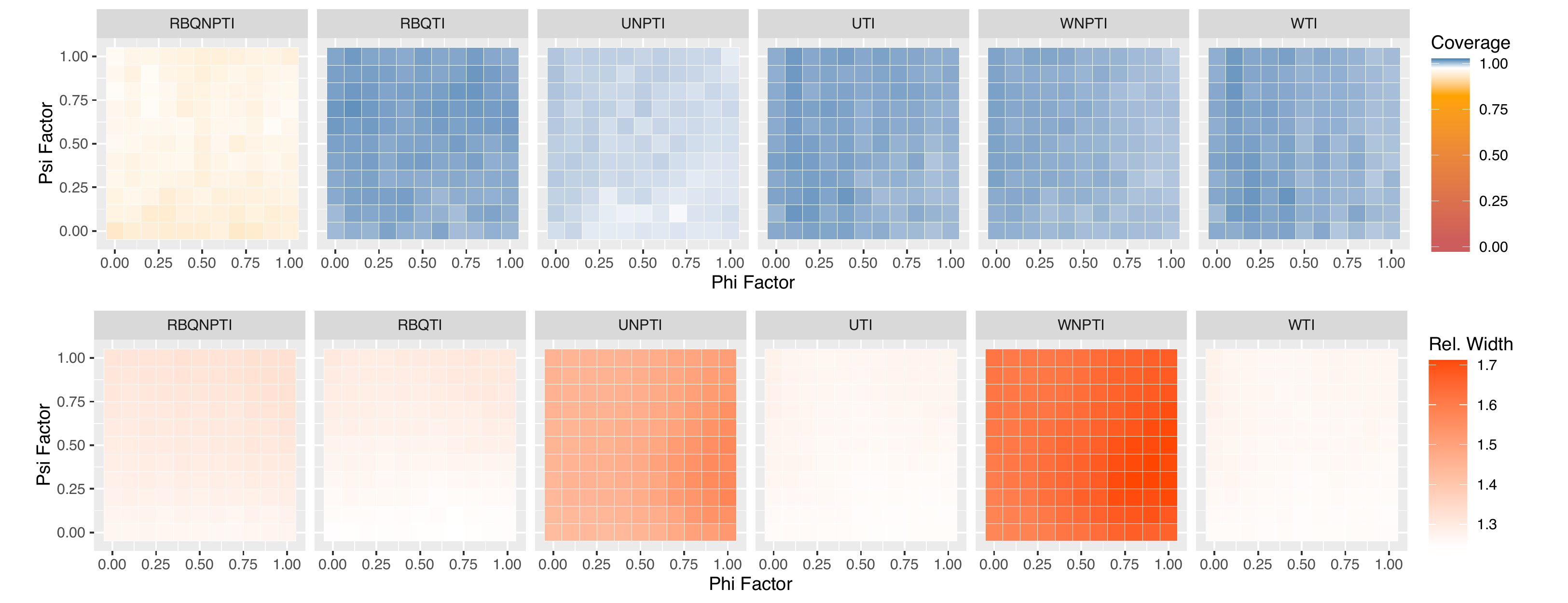}
\caption{\label{fig:10_T3}{\bf {\em t}-distributed} errors (3 degrees of freedom) ~~~~~~~~~~~\\ ~\\}
\end{subfigure}
\begin{subfigure}{\textwidth}
\includegraphics[width=\textwidth]{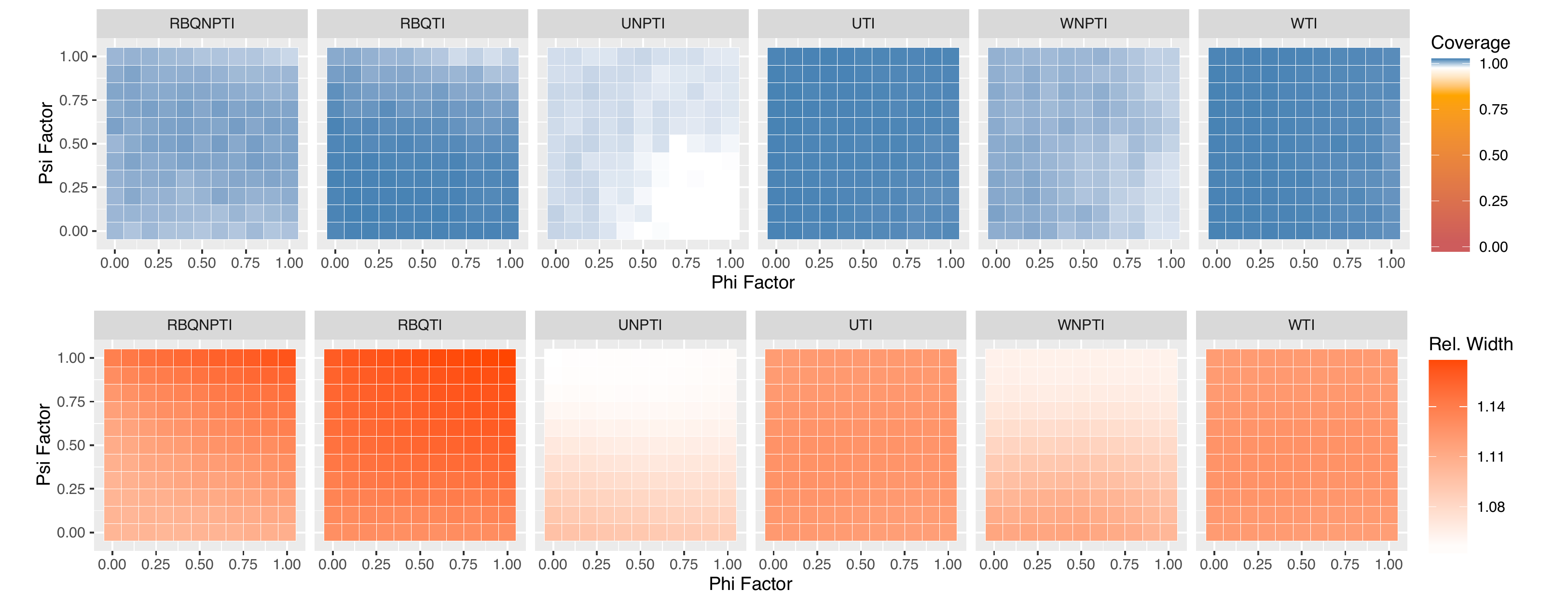}
\caption{\label{fig:10_UNIF}{\bf uniform} errors ~~~~~~~~~~~\\ ~\\}
\end{subfigure}
\caption{\label{fig:10_ALL}Coverage and average Relative Width for all methods and $\sigma^2_\varepsilon ={10}$. Phi Factor ($\phif$ in the text) controls the effect of covariates on the exploration DTR, and Psi Factor ($\psif$ in the text) controls the effect of $A_2$ on the conditional mean of $Y$.}
\end{figure*}
\clearpage

\begin{figure*}
\begin{subfigure}{\textwidth}
\includegraphics[width=\textwidth]{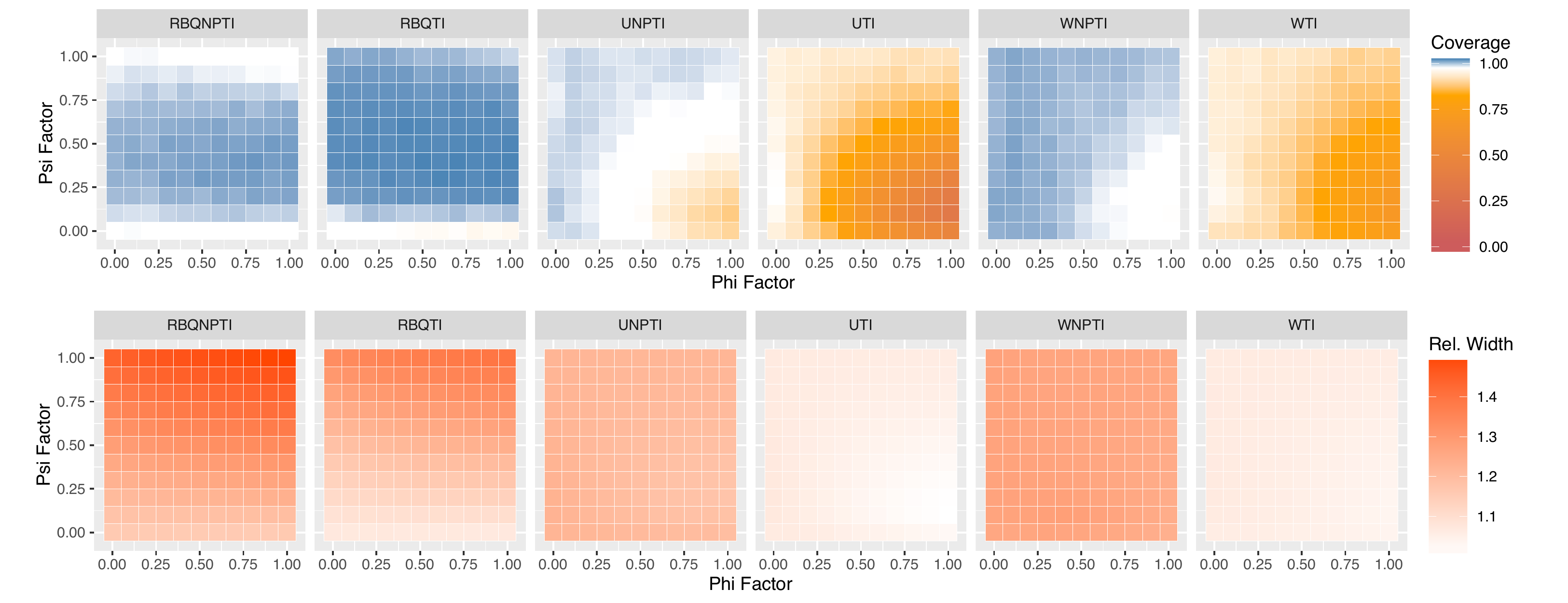}
\caption{\label{fig:1_NORM}{\bf normal} errors ~~~~~~~~~~~\\ ~\\}
\end{subfigure}
\begin{subfigure}{\textwidth}
\includegraphics[width=\textwidth]{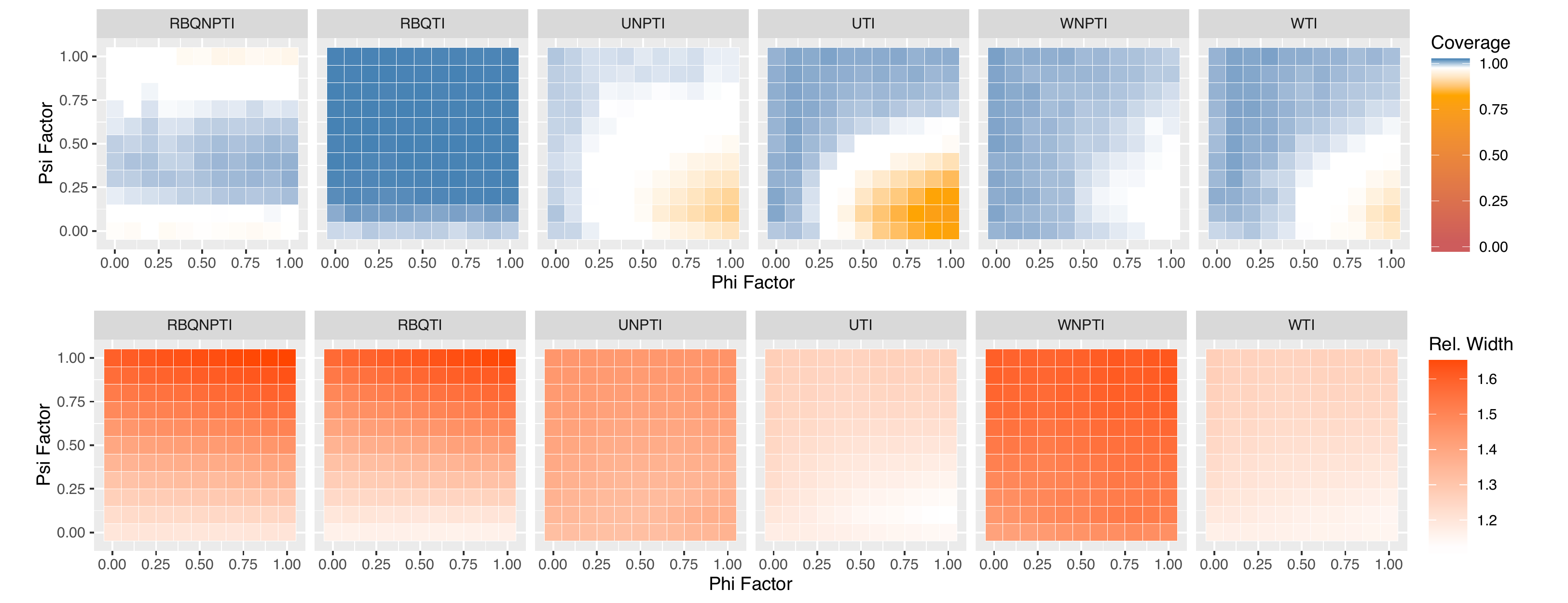}
\caption{\label{fig:1_T3}{\bf {\em t}-distributed} errors (3 degrees of freedom) ~~~~~~~~~~~\\ ~\\}
\end{subfigure}
\begin{subfigure}{\textwidth}
\includegraphics[width=\textwidth]{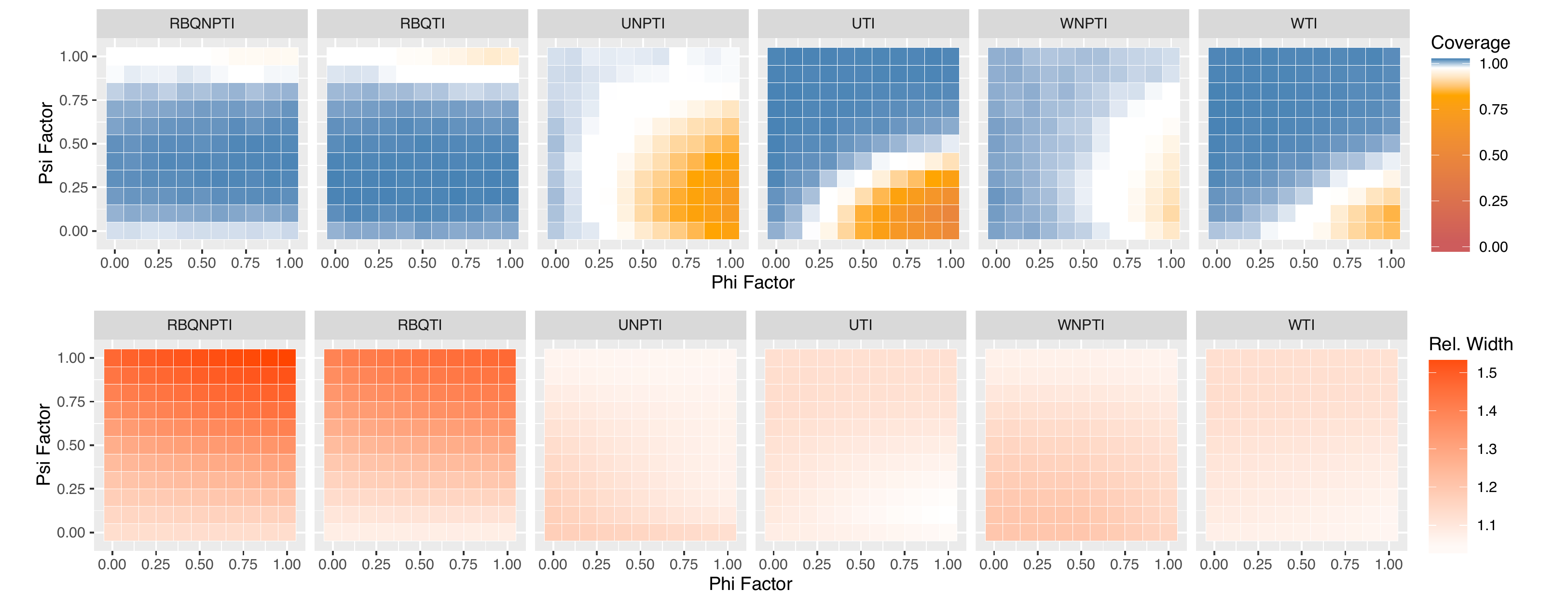}
\caption{\label{fig:1_UNIF}{\bf uniform} errors ~~~~~~~~~~~\\ ~\\}
\end{subfigure}
\caption{\label{fig:1_ALL}Coverage and average Relative Width for all methods and $\sigma^2_\varepsilon ={1}$. Phi Factor ($\phif$ in the text) controls the effect of covariates on the exploration DTR, and Psi Factor ($\psif$ in the text) controls the effect of $A_2$ on the conditional mean of $Y$.}
\end{figure*}
\clearpage

\begin{figure*}
\begin{subfigure}{\textwidth}
\includegraphics[width=\textwidth]{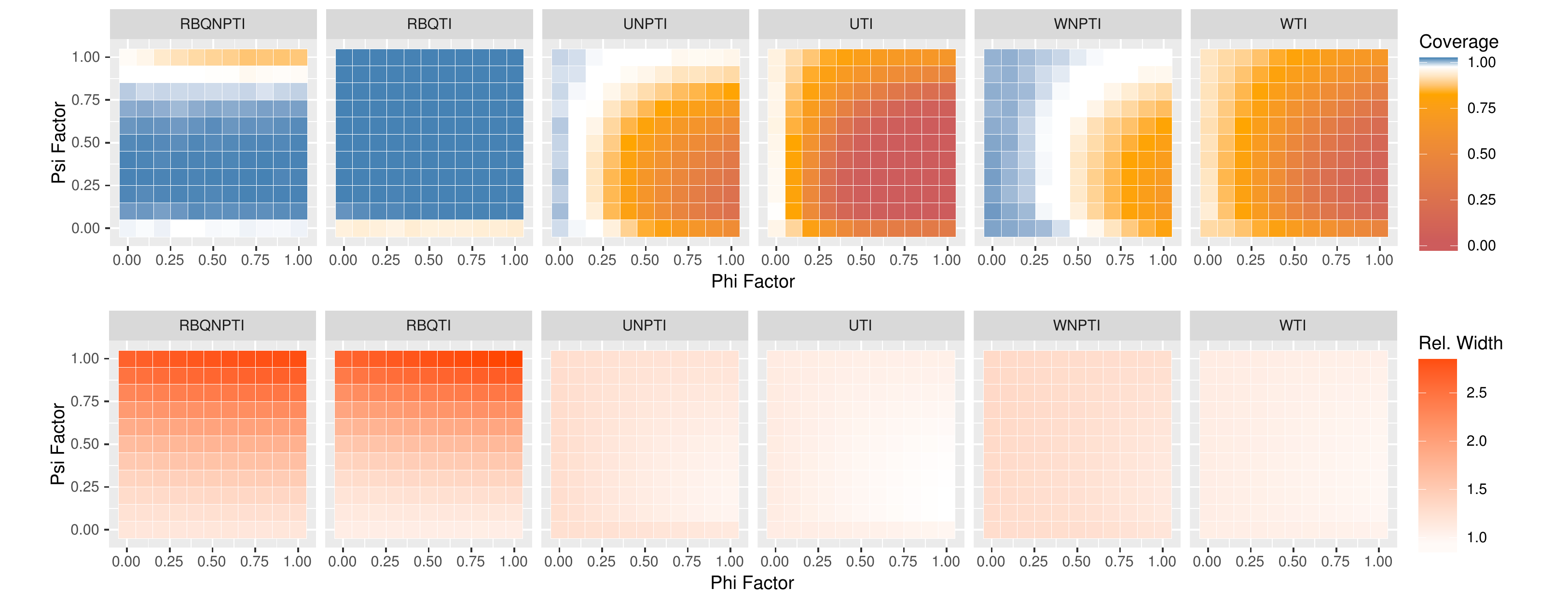}
\caption{\label{fig:.1_NORM}{\bf normal} errors ~~~~~~~~~~~\\ ~\\}
\end{subfigure}
\begin{subfigure}{\textwidth}
\includegraphics[width=\textwidth]{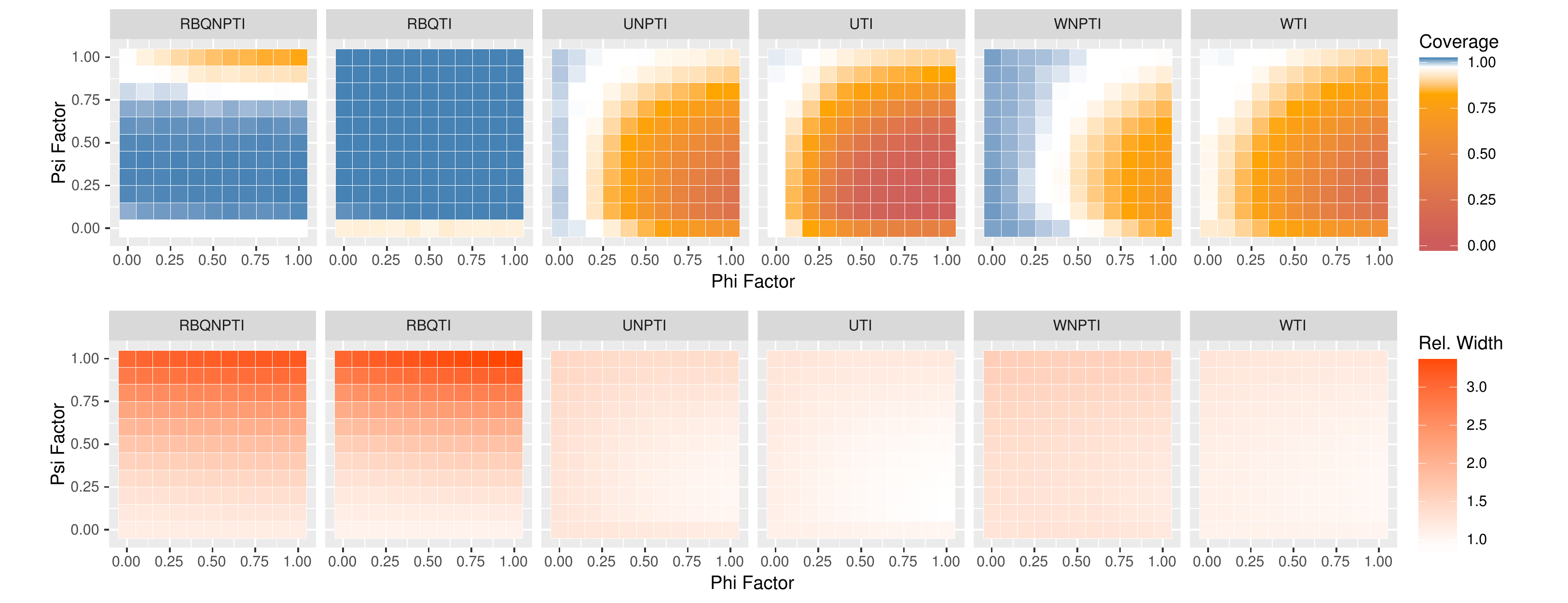}
\caption{\label{fig:.1_T3}{\bf {\em t}-distributed} errors (3 degrees of freedom) ~~~~~~~~~~~\\ ~\\}
\end{subfigure}
\begin{subfigure}{\textwidth}
\includegraphics[width=\textwidth]{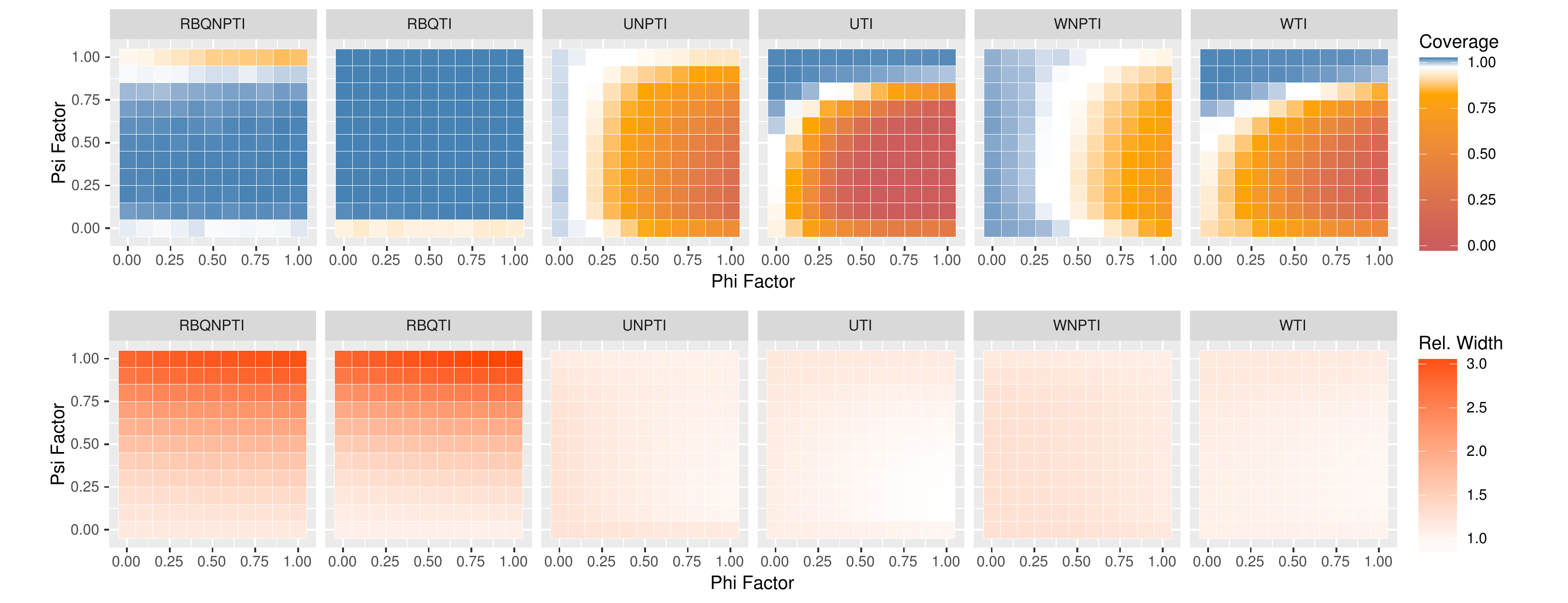}
\caption{\label{fig:.1_UNIF}{\bf uniform} errors ~~~~~~~~~~~\\ ~\\}
\end{subfigure}
\caption{\label{fig:.1_ALL}Coverage and average Relative Width for all methods and $\sigma^2_\varepsilon ={.1}$. Phi Factor ($\phif$ in the text) controls the effect of covariates on the exploration DTR, and Psi Factor ($\psif$ in the text) controls the effect of $A_2$ on the conditional mean of $Y$.}
\end{figure*}

\end{document}